\documentclass[12pt,reqno,fleqn]{amsart}

\usepackage{amssymb,amsmath,amsthm,amscd,mathrsfs,amsbsy,amsfonts}
\usepackage{pstricks,pst-node}
\usepackage{amsbsy,young,colortbl,cite}

\usepackage{graphicx}

\usepackage{amsmath}

\usepackage{amssymb,amsthm}
\usepackage[english]{babel}
\renewcommand{\d}{\partial}

\newtheorem{remark}{Remark}
\newtheorem{proposition}{Proposition}
\newtheorem{lemma}{Lemma}
\newtheorem{definition}{Definition}
\newtheorem{theorem}{Theorem}

\makeatletter      
\@addtoreset{equation}{section}

\newcommand{\I}{\mathbb{I}}
\newcommand{\bS}{\mathbb{S}}
\renewcommand{\d}{\mathrm{d}}

\newcommand{\Exp}[1]{\operatorname{e}^{#1}}

\newcommand{\g}{\mathfrak{g}}

\newcommand{\Cc}{\mathcal{C}}

\renewcommand{\L}{\mathcal L}

\newcommand{\W}{\mathcal W}

\newcommand{\Z}{\mathbb Z}
\newcommand{\C}{\mathbb C}
\newcommand{\N}{\mathbb N}
\def\res{\mathop{\rm Res}\nolimits}

\renewcommand{\L}{\mathcal{L}}

\makeatletter
    
    \newcommand{\Rmnum}[1]{\expandafter\@slowromancap\romannumeral #1@}
  \makeatother
\def\res{\mathop{\rm Res}\nolimits}
\def\({\left(}
\def\){\right)}
\def\[{\begin{eqnarray}}
\def\]{\end{eqnarray}}
\def\d{\partial}

\newcommand{\La}{\Lambda}
\begin{document}

\title{Sato theory on the $q$-Toda hierarchy and its extension}

\author{
Chuanzhong Li\  } \dedicatory {  Department of Mathematics and Ningbo Collabrative Innovation Center of Nonlinear Harzard System of Ocean and Atmosphere,\\
 Ningbo university, Ningbo 315211, China,\\
 lichuanzhong@nbu.edu.cn}

\date{}

\maketitle

\begin{abstract}
In this paper, we construct the Sato theory including the Hirota bilinear equations and tau function of a new  $q$-deformed Toda hierarchy(QTH). Meanwhile the Block type additional symmetry and bi-Hamiltonian structure of this hierarchy are given. From  Hamiltonian tau symmetry, we give another definition of tau function of this hierarchy. Afterwards, we extend the $q$-Toda hierarchy to an extended $q$-Toda hierarchy(EQTH) which satisfy a generalized Hirota quadratic equation in terms of generalized vertex operators. The Hirota quadratic equation might have further application in Gromov-Witten theory. The corresponding Sato theory including multi-fold Darboux transformations of this extended hierarchy is also constructed. At last, we  construct the multicomponent extension of the $q$-Toda hierarchy and show the integrability including its  bi-Hamiltonian structure, tau symmetry and conserved densities.
\end{abstract}


\tableofcontents

\section {Introduction}

The Toda lattice and KP hierarchy are  completely integrable systems which have many important applications in mathematics and physics including the theory of Lie algebra representation, orthogonal polynomials and  random
matrix model  \cite{Toda,Todabook,UT,witten,dubrovin}. KP and Toda systems have many kinds of reduction or extension, for example BKP, CKP hierarchy, extended Toda hierarchy (ETH)\cite{CDZ,M}, bigraded Toda hierarchy (BTH)\cite{C}-\cite{ourBlock} and so on.

The $q$-calculus ( also called quantum calculus)  traces
back to the early 20th century. Many mathematicians have  important
works in the area of $q$-calculus  and $q$-hypergeometric series \cite{qHyper,qseries}. The
$q$-deformation of  classical nonlinear integrable system  started in 1990's by means of
$q$-derivative $\partial_q $   instead of usual derivative with respect to $x$  in the classical system. As we know, the
$q$-deformed integrable system reduces to a classical integrable
system when $q$ goes to 1.

Several $q$-deformed integrable systems have been presented, for example the
$q$-deformed Kadomtsev-Petviashvili ($q$-KP) hierarchy is a
subject of intensive study in the literature \cite{mas}-\cite{hetianqkp}.
 Basing on a similar $q$-operator as $q$-KP hierarchy in \cite{frenkel,iliev}, the $q$-Toda equation was studied in \cite{ZTAK,Silindir} but not for a whole hierarchy. This paper will be devoted to the further studies on the whole $q$-Toda hierarchy(QTH) and its extended hierarchy with logarithmic flows.

Adding additional logarithmic flows to the Toda lattice hierarchy,
it becomes the extended Toda hierarchy\cite{CDZ} which governs the Gromov-Witten invariant of $CP^1$. Therefore what is the application in Gromov-Witten theory of the q-deformed extended Toda hierarchy becomes a natural question which is one motivation for us to do this work.
 The extended bigraded Toda
hierarchy(EBTH)\cite{ourJMP} is the extension of the bigraded Toda
hierarchy (BTH) which includes additional logarithmic flows\cite{C,KodamaCMP}. The Hirota bilinear equation of the EBTH was equivalently constructed in our early paper. One can also consider the bigraded extension of the extended QTH which might be included in our future work.

The multicomponent 2D Toda hierarchy
 was considered from the point of view of the Gauss-Borel factorization problem, non-intersecting Brownian motions and matrix Riemann-Hilbert problem \cite{manasInverse2}-\cite{manas}. In fact the multicomponent 2D Toda hierarchy in \cite{manasinverse} is a periodic reduction of the bi-infinite matrix-formed two dimensional Toda hierarchy. The coefficients of the multicomponent 2D Toda hierarchy take values in complex finite-sized
matrices. In this paper, we also construct the multicomponent extension of the $q$-Toda hierarchy and show the integrability including its  bi-Hamiltonian structure, tau symmetry.

This paper is arranged as follows. In the next section we recall a factorization problem and construct the Lax  equations of the  $q$-Toda hierarchy. In Section 3-7,
we will give the Sato theory of the $q$-Toda hierarchy (QTH) including Hirota bilinear equations, the tau function, vertex operators and Hirota quadratic equations.
Basing on the double dressing structure of this hierarchy, the Block type Lie symmetry \cite{ourBlock,torus} of the QTH was given in Section 8. In Section 9-11, we generalize the Sato theory of the $q$-Toda hierarchy to the extended $q$-Toda hierarchy(EQTH). To prove the integrability of this new extended hierarchy, the bi-Hamiltonian structure and tau symmetry of the EQTH are constructed. In Section 12, the multi-fold Darboux transformation of the EQTH was given which can produce new solutions from seed solutions as used in \cite{Hedeterminant,rogueHMB,EMTH,EZTH}. In Section 13-15, we  construct the multicomponent extension of the $q$-Toda hierarchy and show the integrability including the  bi-Hamiltonian structure, tau symmetry and conserved densities of this matrix hierarchy.

\section{Factorization and dressing operators}

Now we will consider the the shift operator $\Lambda_q$ acting on these functions as
$(\Lambda_q g)(x):=g(qx)$, i.e. $\Lambda_q :=e^{\epsilon x\d_x}, q=e^{\epsilon}$. A Left
multiplication by  $X$ is as $X\Lambda_q^j$, $(
X\Lambda_q^j)(g)(x):=X(x)\circ g(q^jx)$ with defining the product
$(X(x)\Lambda_q^i)\circ(Y(x)\Lambda_q^j):=X(x)Y(q^ix)\Lambda_q^{i+j}.$

The Lie algebra\begin{align*}
  \g&=\Big\{\sum_{j}X_j(x)\Lambda_q^j\Big\},
\end{align*} has the following important splitting
\begin{gather}\label{splitting}
\g=\g_+\circ\g_-,
\end{gather}
where
\begin{align*}
  \g_+&=\Big\{\sum_{j\geq 0}X_j(x)\Lambda_q^j,\Big\},&
  \g_-&=\Big\{\sum_{j< 0}X_j(x)\Lambda_q^j,\Big\}.
\end{align*}

For the corresponding Lie group $G$ whose Lie algebra is $\g$, the splitting
\eqref{splitting} leads us to consider the following factorization of
$g\in G$
\begin{gather}\label{fac1}
g=g_-^{-1}\circ g_+, \quad g_\pm\in G_\pm
\end{gather}
where $G_\pm$ have $\g_\pm$ as their Lie algebras. $G_+$
is the set of invertible linear operators  of the
form $\sum_{j\geq 0}g_j(x)\Lambda_q^j$; while $G_-$ is the set of
invertible linear operators of the form
$1+\sum_{j<0}g_j(x)\Lambda_q^j$.
Then the set $\g$ of Laurent
series in $\Lambda_q$ as an associative algebra is a
Lie algebra under the standard commutator. Similar as \cite{Bergvelt}, the factorization \eqref{fac1}  belong to the big cell [4] and the
factorization is defined only locally to avoid the generation of additional problems connected
with these local aspects.

 Now we
introduce  the following free operators $ W_0,\bar  W_0\in G$
\begin{align}
 \label{def:E}  W_0&:=\Exp{\sum_{j=0}^\infty
 t_{j}\frac{\Lambda_q^j}{\epsilon j!}}, \\
\label{def:barE}   \bar W_0&:=\Exp{\sum_{j=0}^\infty
   t_{j}\frac{\Lambda_q^{-j}}{\epsilon j!}},
\end{align}
where $t_{j}\in \C$
will play the role of continuous times.

 We   define the dressing operators $W,\bar W$ as follows
\begin{align}
\label{def:baker}W&:=S\circ W_0,\ \  \bar W:=\bar S\circ \bar  W_0,\quad S\in G_{-},\ \bar S\in G_{+}.
\end{align}
Given an element $g\in G$ and denote $t=(t_{j}),  j\mathbb\in \N$, one can consider the factorization problem  in $G$.
\begin{gather}
  \label{facW}
  W\circ g=\bar W,
\end{gather}
i.e.
 the factorization problem
\begin{gather}
  \label{factorization}
  S(t)\circ W_0\circ g=\bar S(t)\circ\bar W_0.
\end{gather}
Observe that  $S,\bar S$ have expansions of the form
\begin{gather}
\label{expansion-S}
\begin{aligned}
S&=1+\omega_1(x)\Lambda_q^{-1}+\omega_2(x)\Lambda_q^{-2}+\cdots\in G_{-},\\
\bar S&=\bar\omega_0(x)+\bar\omega_1(x)\Lambda_q+\bar\omega_2(x)\Lambda_q^{2}+\cdots\in
G_{+}.
\end{aligned}
\end{gather}
Also we define the symbols of $S,\bar S$ as  $\bS,\bar \bS$
\begin{gather}
\begin{aligned}
\bS&=1+\omega_1(x)\lambda^{-1}+\omega_2(x)\lambda^{-2}+\cdots,\\
\bar \bS&=\bar\omega_0(x)+\bar\omega_1(x)\lambda+\bar\omega_2(x)\lambda^{2}+\cdots.
\end{aligned}
\end{gather}

The inverse operators $S^{-1},\bar S^{-1}$ of operators $S,\bar S$ have expansions of the form
\begin{gather}
\begin{aligned}
S^{-1}&=1+\omega'_1(x)\Lambda_q^{-1}+\omega'_2(x)\Lambda_q^{-2}+\cdots\in G_{-},\\
\bar S^{-1}&=\bar\omega'_0(x)+\bar\omega'_1(x)\Lambda_q+\bar\omega'_2(x)\Lambda_q^{2}+\cdots\in
G_{+}.
\end{aligned}
\end{gather}
Also we define the symbols of $S^{-1},\bar S^{-1}$  as  $\bS^{-1},\bar \bS^{-1}$ as following
\begin{gather}
\begin{aligned}
\bS^{-1}&=1+\omega'_1(x)\lambda^{-1}+\omega'_2(x)\lambda^{-2}+\cdots,\\
\bar \bS^{-1}&=\bar\omega'_0(x)+\bar\omega'_1(x)\lambda+\bar\omega'_2(x)\lambda^{2}+\cdots.
\end{aligned}
\end{gather}

 The Lax  operators $\L\in G$ of the $q$-deformed Toda hierarchy
 are defined by
\begin{align}
\label{Lax}  \L&:=W\circ\Lambda_q\circ W^{-1}=\bar W\circ\Lambda_q^{-1}\circ \bar W^{-1},
\end{align}
and
have the following expansions
\begin{gather}\label{lax expansion}
\begin{aligned}
 \L&=\Lambda_q+U(x)+V(x)\Lambda_q^{-1}.
\end{aligned}
\end{gather}
 In fact the Lax  operators $\L\in G$
 are also be equivalently defined by
\begin{align}
\label{two dressing}  \L&:=S\circ\Lambda_q\circ S^{-1}=\bar S\circ\Lambda_q^{-1}\circ \bar S^{-1}.
\end{align}

\section{ Lax equations of QTH}

In this section we will use the factorization problem \eqref{facW} to derive  Lax equations.
Let us first introduce some convenient notation on the operators $B_{j}$ defined as follows
\begin{align}\label{satoS}
\begin{aligned}
B_{j}&:=\frac{\L^{j+1}}{(j+1)!}.
\end{aligned}
\end{align}

Now we give the definition of the  $q$-Toda hierarchy(QTH).
\begin{definition}The  $q$-Toda hierarchy is a hierarchy in which the dressing operators $S,\bar S$ satisfy following Sato equations
\begin{align}
\label{satoSt} \epsilon\partial_{t_{j}}S&=-(B_{j})_-S,& \epsilon\partial_{t_{j}}\bar S&=(B_{j})_+\bar S.\end{align}
\end{definition}
Then one can easily get the following proposition about $W,\bar W.$

\begin{proposition}The dressing operators $W,\bar W$ are subject to following Sato equations
\begin{align}
\label{Wjk} \epsilon\partial_{t_{j}}W&=(B_{j})_+ W,& \epsilon\partial_{t_{j}}\bar W&=(B_{j})_+\bar W.  \end{align}
\end{proposition}

 From the previous proposition we derive the following  Lax equations for the Lax operators.
\begin{proposition}\label{Lax}
 The  Lax equations of the QTH are as follows
   \begin{align}
\label{laxtjk}
  \epsilon\partial_{t_{j}} \L&= [(B_{j})_+,\L].
  \end{align}
\end{proposition}

To show the relation of the QTH and the $q$-KP type hierarchy\cite{frenkel,iliev,iliev2,he,hetianqkp}, we will do the following remark.
\begin{remark}
 The $q$-Toda hierarchy can be treated as a generalization of the $q$-KdV hierarchy \cite{frenkel} in terms of the same multiplication shift operator $\Lambda_q$. The $q$-KP hierarchy in \cite{iliev} is in fact a more general generalization of the $q$-KdV hierarchy in \cite{frenkel} after rewriting the operator  $\Delta_q$ as $\Lambda_q-1$. Therefore the $q$-Toda hierarchy can be treated as a special reduction of the $q$-KP hierarchy in \cite{iliev} in terms of an operator $\Delta_q=\Lambda_q-1$ after a certain transformation. The operator $\Lambda_q$ in this paper is different from the $q$-derivative operator in \cite{iliev2,he,hetianqkp} in which $D_qf(x)=\frac{f(qx)-f(x)}{(q-1)x}$ which leads to a different hierarchy.
\end{remark}

To see this kind of hierarchy more clearly, the  $q$-Toda equations as the $t_{0}$ flow equations  will be given in the next subsection.
\subsection{The  $q$-Toda equations}
 As a consequence of the factorization problem \eqref{facW} and  Sato equations, after taking into account that   $S\in G_{-}$ and $\bar S\in G_{+}$, the $t_0$ flow of $\L$ in the form of $\L=\Lambda_q+U+V\Lambda_q^{-1}$ is as
\begin{gather}\label{exp-omega}
\begin{aligned}
  \epsilon\partial_{t_{0}} \L&= [\Lambda_q+U,V\Lambda_q^{-1}],
  \end{aligned}
\end{gather}
which lead to $q$-Toda equation
\[\epsilon\partial_{t_{0}} U&=& V(qx)-V(x),\\ \label{toda}
\epsilon\partial_{t_{0}} V&=& U(x)V(x)-V(x)U(q^{-1}x).\]

From Sato equation we deduce the following set of nonlinear
partial differential-difference equations
\begin{align}\left\{
\begin{aligned}
 \omega_1(x)-\omega_1(qx)&=\epsilon\partial_{t_1}(\Exp{\phi(x)})\cdot\Exp{-\phi(x)},\\
\epsilon\partial_{t_1}\omega_1(x)&=-\Exp{\phi(x)}\Exp{-\phi(q^{-1}x)}.\end{aligned}\right.
\label{eq:multitoda}
\end{align}
Observe that if we cross the first two  equations, then we get
\begin{align*}
  \epsilon^2\partial_{t_1}^2\phi(x)=
  \Exp{\phi(qx)}\Exp{-\phi(x)}-\Exp{\phi(x)}\Exp{-\phi(q^{-1}x)}
\end{align*}
which is the $q$-Toda equation.
To give a linear description of the QTH, we introduce  wave functions  $\psi,\bar\psi$
defined by
\begin{gather}\label{baker-fac}
\begin{aligned}
\psi&= W\cdot\chi, &
\bar\psi&=\bar W\cdot \bar\chi,
\end{aligned}
\end{gather}
where
\[
\chi(z):=z^{\frac{\log x}{ \epsilon}} ,\ \ \bar \chi(z):=z^{-\frac{\log x}{\epsilon}},\
\]
and the $``\cdot"$ means the action of an operator on a function.
Note that $\Lambda_q\cdot\chi=z\chi$ and  the following asymptotic expansions
can be defined
\begin{gather}\label{baker-asymp}
\begin{aligned}
  \psi&=z^{\frac{\log x}{\epsilon}}(1+\omega_1(x)z^{-1}+\cdots)\,\psi_0(z),&\psi_0&:=
 \Exp{\sum_{j=1}^\infty t_{j}\frac{z^j}{\epsilon j!}},& z&\rightarrow\infty,\\
\bar\psi&=z^{-\frac{\log x}{\epsilon}}(\bar\omega_0(x)+\bar\omega_1(x)z+\cdots)\,\bar\psi_0(z),
&\bar\psi_0&:=
\Exp{\sum_{j=0}^\infty
   t_{j}\frac{z^{-j}}{\epsilon j!}},& z&\rightarrow 0.
\end{aligned}
\end{gather}

We can further get linear equations in the following proposition.

\begin{proposition}The  wave functions $\psi,\bar\psi$ are subject to following Sato equations
\begin{align}
 \L\cdot\psi&=z\psi,\ \ \ &&\L\cdot\bar\psi=z\bar\psi,\\
 \epsilon\partial_{t_j}\psi&=(B_{j})_+\cdot \psi,& \epsilon\partial_{t_j}\bar \psi&=(B_{j})_+\cdot\bar \psi.  \end{align}
\end{proposition}

\section{Hirota bilinear equations of the QTH}

Basing on above, Hirota bilinear equations which are equivalent to Lax equations of the QTH can be derived in following proposition.
\begin{proposition}\label{HBEoper}
 $W$ and $\bar W$ are  wave operators of the  $q$-Toda hierarchy if and only the following Hirota bilinear equations hold
\begin{align}
W\Lambda_q^r W^{-1}&=\bar W\Lambda_q^{-r}\bar W^{-1}, \ r\in \N.
   \end{align}
\end{proposition}

\begin{proof}
The proof is complicated but quite standard. One can refer the similar proofs in \cite{M,ourJMP}.
\end{proof}

To give a description in terms of  wave functions, following symbolic definitions are needed.

If the series have forms
\begin{eqnarray*} W(x,t,\Lambda_q)=\sum_{i\in \Z} a_i(x,t)\Lambda_q^i \mbox{ and } \bar W(x,t,\Lambda_q)=\sum_{i\in \Z}
b_i(x,t)\Lambda_q^{i}, \end{eqnarray*}

\begin{eqnarray*}  W^{-1}(x,t,\Lambda_q)=\sum_{i\in \Z}\Lambda_q^{i} a_i'(x,t) \mbox{ and } \bar W^{-1}(x,t,\Lambda_q)=\sum_{j\in
\Z}\Lambda_q^{j}b_j'(x,t),  \end{eqnarray*} then we denote their corresponding
left symbols $\W$,  $\bar \W$ and right symbols $\W^{-1}$, $\bar \W^{-1}$
as following
\begin{eqnarray*}
&&  \W(x,t,\lambda) =\sum_{i\in \Z} a_i(x,t)\lambda^i,\ \  \W^{-1}(x,t,\lambda)=  \sum_{i\in \Z} a_i'(x,t)\lambda^{i},\\
&&
\bar \W(x,t,\lambda) =\sum_{i\in \Z} b_i(x,t)\lambda^{i},\ \ \bar \W^{-1}(x,t,\lambda)=\sum_{j\in
\Z}b_j'(x,t)\lambda^{j}.
\end{eqnarray*}
With above preparation, it is time to give another form of Hirota bilinear equation(see following proposition) after defining residue as $\res_{\lambda }\sum_{n\in \Z}\alpha_n \lambda^n=\alpha_{-1}$ using the similar proof as \cite{UT,M,ourJMP}.
\begin{proposition}\label{wave-operators}
Let
 $S$ and
$\bar S$ are  wave operators of the $q$-Toda hierarchy if and only if for all   $m\in
\Z$, $r\in \N$ , the following Hirota bilinear identity hold

\begin{eqnarray}  \notag &&\res_{\lambda }
 \left\{
\lambda^{r+m-1}\ \W(x,t,\lambda) \W^{-1}(q^{-m}x,t',\lambda)
\right\} = \\ \label{HBE3}&& \res_{\lambda }
 \left\{
\lambda^{-r+m-1}\bar \W( x,t,\lambda )\
\bar \W^{-1}(q^{-m}x,t',\lambda) \right\}.
\end{eqnarray}
\end{proposition}

To give Hirota quadratic function in terms of tau functions, we need to define and prove the existence of the tau function of the QTH firstly in the next section.

\section{Tau-functions of QTH}

We firstly introduce the following sequences:
\[t-[\lambda] &:=& (t_{j}-
  \epsilon(j-1)!\lambda^j, 0\leq j\leq \infty).
\]
A  function $\tau\in \C$  depending on the dynamical variables $t$ and
$\epsilon$ is called the    tau-function of the QTH if it
provides symbols related to  wave operators as following,

\begin{eqnarray}\label{Mpltaukk}\bS: &=&\frac{ \tau
(e^{-\frac{\epsilon}{2}}x, t_{j}-\frac{\epsilon(j-1)!}{\lambda^j},\epsilon) }
     {\tau (e^{-\frac{\epsilon}{2}}x,t,\epsilon)},\\
     \label{Mpl-1taukk}\bS^{-1}: &=&\frac{ \tau
(e^{\frac{\epsilon}{2}}x, t_{j}+\frac{\epsilon(j-1)!}{\lambda^j},\epsilon) }
     {\tau (e^{\frac{\epsilon}{2}}x,t,\epsilon)},\\ \label{Mprtaukk}
\bar \bS:&= &\frac{ \tau
(e^{\frac{\epsilon}{2}}x,t_{j}+\epsilon(j-1)!\lambda^j,\epsilon)}
     {\tau(e^{-\frac{\epsilon}{2}}x,t,\epsilon)},\\
     \bar \bS^{-1}:&= &\frac{\tau
(e^{-\frac{\epsilon}{2}}x,t_{j}-\epsilon(j-1)!\lambda^j,\epsilon)}
     {\tau(e^{\frac{\epsilon}{2}}x,t,\epsilon)}.
     \end{eqnarray}

\begin{proposition}\label{tau-function}
Given a pair of wave operators $\bS$ and $\ \bar \bS$ of the QTH, there
exists corresponding  invertible tau-functions.
\end{proposition}
\begin{proof} Here, we shall note that the  tau function $\tau(x,t)$
corresponding to the wave operators $\bS$ and  $\ \bar \bS$ is in fact
$\tau(q^{-\frac12}x,t)$.
\\The system is equivalent to:
\begin{eqnarray*}
&& \label{eq1} \log \bS = \(\exp\left({-\epsilon\sum_{j=0}^\infty j!\lambda^{-(j+1)}\partial_{t_{j}}}\right)-1\)\log \tau, \\
\label{eq2} && \log  \bar \bS = \(\exp\left(\epsilon x\d_{x}+\epsilon\sum_{j=0}^\infty j!\lambda^{j+1}\partial_{t_{j}}\right)-1\)\log \tau.
\end{eqnarray*}
Then using the standard method in \cite{M,ourJMP} will help us to derive the existence of tau function of this hierarchy.

 \end{proof}

After giving tau functions of the QTH, what is the Hirota bilinear equation in terms of the tau function becomes a natural question which will be answered in the next
section in terms of vertex operators.

\section{Vertex operators and  Hirota quadratic equations}
In this section we continue to  discuss on the fundamental properties
of the tau function of the QTH, i.e., the Hirota quadratic equations of the QTH. So we
introduce the following vertex operators
\begin{eqnarray*}
\Gamma^{\pm a} :&=&\exp\left(\pm \frac{1}{\epsilon}
\sum_{j=0}^\infty t_{j}\frac{\lambda^{j+1}}{(j+1)!}\right)\times\exp\left({\mp
\frac{\epsilon}{2}x\partial_{x} \mp [\lambda^{-1}]_\d  }\right),\\
\Gamma^{\pm b} :&=&\exp\left(\pm  \frac{1}{\epsilon}
\sum_{j=0}^\infty t_{j}\frac{\lambda^{-j-1}}{(j+1)!}\right)\times\exp\left({\mp
\frac{\epsilon}{2}x\partial_{x} \mp [\lambda]_{\d}  }\right),
\end{eqnarray*}

where\ \
\begin{eqnarray*}
 [\lambda]_\d  :&=&
\epsilon\sum_{j=0}^\infty j!\lambda^{j+1}\partial_{t_{j}}.
\end{eqnarray*}

\begin{theorem}\label{t11}
The invertible   $\tau(t,\epsilon)$  is a tau-function of the QTH if and only if it
satisfies the  following Hirota quadratic equations  of the QTH.
\begin{equation} \label{HBE} \res_{{\rm{\lambda}}}
 \lambda^{r-1} \( \Gamma^{a}\otimes
\Gamma^{-a}\right)(\tau
\otimes \tau ) =\res_{{\rm{\lambda}}}\lambda^{-r-1}\(\Gamma^{-b}\otimes\Gamma^{b} \) (\tau
\otimes \tau )
\end{equation}
computed at $x=q^lx'$
 for each  $l\in \Z$, $r\in \N$.
\end{theorem}
 \begin{proof}

 We just need  to prove that the Hirota quadratic equations
are equivalent to the right side in Proposition
\ref{wave-operators}. By a straightforward computation we can get
the following four identities {\allowdisplaybreaks}
\begin{eqnarray}\label{vertex computation1}
 \Gamma^{a}\tau & =& \tau(q^{-\frac12}x,t)
 \W(x,t,\lambda )\lambda^{\log x/\epsilon},
\\ \label{vertex computation2}
\Gamma^{-a}\tau  & =&
\lambda ^{-\log x/\epsilon}
\W^{-1}(x,t,\lambda )\tau(q^{\frac12}x,t), \\\label{vertex
computation3}
 \Gamma^{-b}\tau  & =&
\tau(q^{-\frac12}x,t) \bar \W(x,t,\lambda
)\lambda^{ x/\epsilon}, \\\label{vertex computation4}
\Gamma^{b}\bar\tau & = & \lambda^{
 -\log x/\epsilon} \bar \W^{-1}(x,t,\lambda)\
\tau(q^{\frac12}x,t) .
\end{eqnarray}
The proof of four equations eq.\eqref{vertex
computation1}-eq.\eqref{vertex computation4} can be derived by similar QTH as in \cite{M,ourJMP}.
By substituting four equations eq.\eqref{vertex
computation1}-eq.\eqref{vertex computation4} into the Hirota quadratic equations
\eqref{HBE},
eq.\eqref{HBE3} is derived.
\end{proof}

 Doing a transformation on the eq.\eqref{HBE} by $\lambda\rightarrow \lambda^{-1},$ then the eq.\eqref{HBE} becomes
\begin{equation} \label{HBE'} \res_{{\rm{\lambda}}}
 \lambda^{r-1}\left( \Gamma^{a}\otimes
\Gamma^{-a}-\Gamma^{-a}\otimes\Gamma^{a}  \right)(\tau
\otimes \tau )=0
\end{equation}
computed at $x=q^lx'$
 for each  $l\in \Z$, $r\in \N$.
 That means
 \begin{equation} \label{HBEmilanov}\frac{d \lambda}{\lambda} \left(  \Gamma^{a}\otimes
\Gamma^{-a}-\Gamma^{-a}\otimes\Gamma^{a} \right)(\tau
\otimes \tau )
\end{equation}
is regular in $\lambda$
computed at $x=q^lx'$
 for each  $l\in \Z$. The eq.\eqref{HBEmilanov} i is exactly the $q$-version of the
  Hirota quadratic equation of the  Toda hierarchy as a corollary  in \cite{M}.

\section{Bi-Hamiltonian structure and tau symmetry}

To describe the integrability of the QTH, we will construct the bi-Hamiltonian structure and tau symmetry of the QTH in this section.

 In this section, we will consider the QTH on Lax operator
 \[\L=\Lambda_q+u+e^v\Lambda_q^{-1}, \ \ \Lambda_q=e^{\epsilon x\d_x}.\]
Then for $\bar f=\int  f dx, \bar g=\int g dx, $ we can define the hamiltonian bracket as
\[\{\bar f,\bar g\}=\int  \sum_{w,w'}\frac{\delta f}{\delta w}\{w,w'\}\frac{\delta g}{\delta w'} dx,\ \ w,w'=u\ or\ v.\]
The bi-Hamiltonian structure for the
QTH can be given by the following two compatible Poisson brackets similar as \cite{CDZ}

\begin{eqnarray}
&&\{v(x),v(y)\}_1=\{u(x),u(y)\}_1=0,\notag\\
&&\{u(x),v(y)\}_1=\frac{1}{\epsilon} \left[e^{\epsilon\,x\d_x}-1
\right]\delta(x-y),\label{toda-pb1}\\
&& \{u(x),u(y)\}_2={1\over\epsilon}\left[e^{\epsilon\,x\d_x}
e^{v(x)}-
e^{v(x)} e^{-\epsilon x\d_x}\right] \delta(x-y),\notag\\
&& \{ u(x), v(y)\}_2 = {1\over \epsilon}
u(x)\left[e^{\epsilon\,x\d_x}-1 \right]
\delta(x-y),\label{toda-pb2}\\
&& \{ v(x), v(y)\}_2 = {1\over \epsilon} \left[
e^{\epsilon\,x\d_x}-e^{-\epsilon x\d_x}\right]\delta(x-y).\notag
\end{eqnarray}
For any difference operator $A=
\sum_k A_k \Lambda_q^k$, we define residue $Res A=A_0$.
In the following theorem, we will prove the above Poisson structure can be as  the Hamiltonian structure of the QTH.
\begin{theorem}\label{QTHbiha}
The flows of the QTH  are Hamiltonian systems
of the form
\[
\frac{\d u}{\d t_{j}}&=&\{u,H_{j}\}_1, \ \ j\ge 0,
\label{td-ham}
\]

They satisfy the following bi-Hamiltonian recursion relation
\[\notag \{\cdot,H_{n-1}\}_2=n
\{\cdot,H_{n}\}_1.
\]
Here the Hamiltonians have the form
\begin{equation}
H_{j}=\int h_{j}(u,v; u_x,v_x; \dots; \epsilon) dx,\quad  \ j\ge 0,
\end{equation}
with
\[
 h_{j}&=&\frac1{(j+1)!} Res \, \L^{j+1}.
\]

\end{theorem}

\begin{proof}
The proof is similar as the proof in \cite{CDZ}.
Here we will prove that the flows $\frac{\d}{\d t_{n}}$ are also
Hamiltonian systems with respect to the first Poisson bracket.

Suppose
\[
B_{n}=\sum_{k} a_{n+1;k}\, \Lambda_q^k,
\]
and from
\begin{equation}
  \label{edef3}
\frac{\partial \L}{\partial t_{ n}} = [ (B_{n})_+ ,\L ]= [ -(B_{n})_- ,\L ],
\end{equation}
we can derive equation
\[\epsilon\frac{\partial u}{\partial t_{ n}}&=&a_{n+1;1}(qx)-a_{n+1;1}(x),\\
\epsilon\frac{\partial v}{\partial t_{ n}}&=&a_{n+1;0}(q^{-1}x) e^{v(x)}-a_{n+1;0}(x) e^{v(qx)}.
\]

By
\begin{eqnarray}
&&d \tilde h_{n}=\frac1{(n+1)!}\,d\, Res\left[\L^{n+1}
\right]
\notag\\
&& \sim \frac1{n!}\, Res\left[\L^{n}
 d \L\right]\notag\\
&&= Res\left[a_{n;0}(x)du+a_{n;1}(q^{-1}x) e^{v(x)}dv\right],
\end{eqnarray}
it yields the following identities
\begin{equation}\label{dH1-u12}
\frac{\delta H_{n}}{\delta u}=a_{n;0}(x),\quad \frac{\delta H_{n}}
{\delta v}=a_{n;1}(q^{-1}x) e^{v(x)}.
\end{equation}
This agree with Lax equation

\[
\frac{\d u}{\d t_{n}}&=&\{u,H_{n}\}_1={1\over \epsilon} \left[
e^{\epsilon\,x\d_x}-1\right]\frac{\delta H_{n}}
{\delta v}={1\over \epsilon}(a_{n;1}(qx)-a_{n;1}(x)),\\
 \  \frac{\d v}{\d t_{n}}&=&\{v,H_{n}\}_1=\frac{1}{\epsilon} \left[1-e^{\epsilon\,x\d_x}
\right]\frac{\delta H_{n}}
{\delta u}=\frac{1}{\epsilon} \left[a_{n;0}(q^{-1}x) e^{v(x)}-a_{n;0}(x) e^{v(qx)}\right].
\]

 From the above identities we see that
the flows $\frac{\d}{\d t_{n}}$ are Hamiltonian systems
with the first Hamiltonian structure.
The recursion relation
follows from the following trivial identities
\begin{eqnarray}
&&n\, \frac{1}{n!} \L^{n} =\L\,
\frac{1}{(n-1)!}
\L^{n-1}=\frac{1}{(n-1)!}
\L^{n-1}\L.\notag
\end{eqnarray}
Then we get,
\begin{eqnarray}
&&n a_{n;1}(x)=a_{n-1;0}(qx)+ua_{n-1;1}(x)+e^va_{n-1;2}(q^{-1}x)\notag\\
&&=a_{n-1;0}(x)+u(qx)a_{n-1;1}(x)+e^{v(q^2x)}a_{n-1;2}(x).\notag
\end{eqnarray}
This further leads to

\begin{eqnarray}
&&\{u,H_{n-1}\}_2=\{\left[\Lambda_q e^{v(x)}-e^{v(x)} \Lambda_q^{-1}\right] a_{n-1;0}(x)+
u(x) \left[\Lambda_q-1\right] a_{n-1;1}(q^{-1}x) e^{v(x)}\}\notag\\ \notag
&&
=n\left[a_{n;1}(x) e^{v(qx)}-a_{n;1}(q^{-1}x) e^{v(x)}\right].\label{pre-recur}
\end{eqnarray}
This is exactly the recursion relation on flows for $u$. The similar recursion flow on $v$ can be similarly derived.
The theorem is proved till now.

\end{proof}

Similarly as \cite{CDZ}, the tau symmetry of the QTH can be proved in the  following theorem.
\begin{theorem}\label{tausymmetry}
The QTH has the following tau-symmetry property:
\begin{equation}
\frac{\d h_{m}}{\d t_{ n}}=\frac{\d
h_{ n}}{\d t_{m}},\quad \ m,n\ge 0.
\end{equation}
\end{theorem}
\begin{proof} Let us prove the theorem in a direct way
\[
&&\frac{\d h_{m}}{\d t_{n}} =\frac1{m!\,n!}\, Res[-(\L^{n})_-, \L^m]\notag\\
&&=\frac1{m!\,n!}\, Res[(\L^m )_+,(\L^{n})_-]\notag\\
&& =\frac1{m!\,n!}\, Res[(\L^m )_+,\L^{n}]=\frac{\d h_{n}}{\d t_{m}}.
\]
This theorem is proved.
\end{proof}

 This property justifies another alternative definition of the
tau function for the QTH.

\begin{definition} The  $tau$ function $\tau$ of the QTH can also be defined by
the following expressions in terms of the densities of the Hamiltonians:
\begin{equation}
h_{n}=\epsilon (\Lambda_q-1)\frac{\d\log  \tau}{\d t_{n}},
\quad \ n\ge 0.
\end{equation}
\end{definition}

\section{Additional symmetry and Block algebra}

In this section, we will put constrained condition
eq.\eqref{two dressing} into construction of the flows of additional
symmetries which form the well-known Block algebra.

With the dressing operators given in eq.\eqref{two dressing}, we introduce Orlov-Schulman operators as following
\begin{eqnarray}\label{Moperator}
&&M=S\Gamma S^{-1}, \ \ \bar M=\bar S\bar \Gamma \bar S^{-1},\ \\
 &&\Gamma=
\frac{\log x}{\epsilon}\Lambda_q^{-1}+\sum_{n\geq 0}
(n+1)\Lambda_q^{n}t_{n},\ \bar \Gamma=
\frac{-\log x}{\epsilon}\Lambda_q.
\end{eqnarray}

Then one can prove the Lax operator $\L$ and Orlov-Schulman operators $M,\bar M$ satisfy the following proposition.
\begin{proposition}\label{flowsofM}
The Lax operator $\L$ and Orlov-Schulman operators $M,\bar M$ of the QTH
satisfy the following
\begin{eqnarray}
&[\L,M]=1,[\L,\bar M]=1,\\ \label{Mequation}
&\partial_{ t_{n}}M=
[(B_{n})_+,M],\ \ \partial_{ t_{n}}\bar M=[(B_{n})_+,\bar M],\\
&\dfrac{\partial
M^m\L^k}{\partial{t_{n}}}=[(B_{n})_+,
M^m\L^k],\;  \dfrac{\partial
\bar M^m\L^k}{\partial{t_{n}}}=[(B_{n})_+, \bar M^m\L^k]
\end{eqnarray}

\end{proposition}

\begin{proof}
One can prove the proposition by dressing the following several commutative Lie brackets
\begin{eqnarray*}&&[\partial_{ t_{n}}-\frac{\Lambda_q^{n+1}}{(n+1)!},\Gamma]\\
&=&[\partial_{ t_{n}}-\frac{\Lambda_q^{n+1}}{(n+1)!},\frac{\log x}{\epsilon}\Lambda_q^{-1}+\sum_{n\geq 0}
\frac{\Lambda_q^{n}}{n!}t_{n}]\\&=&0,
\end{eqnarray*}

\begin{eqnarray*}[\partial_{ t_{n}},\bar \Gamma]
&=&[\partial_{ t_{n}},\frac{-\log x}{\epsilon}\Lambda_q]=0.
\end{eqnarray*}

\end{proof}
We are now to define the additional flows, and further to prove that they are symmetries, which are called additional
symmetries of the QTH. We introduce additional
independent variables $t^*_{m,l}$ and define the actions of the
additional flows on the wave operators as
\begin{eqnarray}\label{definitionadditionalflowsonphi2}
\dfrac{\partial S}{\partial
{t^*_{m,l}}}=-\left((M-\bar M)^m\L^l\right)_{-}S, \ \ \ \dfrac{\partial
\bar S}{\partial {t^*_{m,l}}}=\left((M-\bar M)^m\L^l\right)_{+}\bar S,
\end{eqnarray}
where $m\geq 0, l\geq 0$. The following theorem shows that the definition \eqref{definitionadditionalflowsonphi2} is compatible with reduction condition \eqref{two dressing} of the  QTH.
\begin{proposition}\label{preserve constraint}
The additional flows \eqref{definitionadditionalflowsonphi2} preserve reduction condition \eqref{two dressing}.
\end{proposition}
\begin{proof} By performing the derivative on $\L$ dressed by $S$ and
using the additional flow about $S$ in \eqref{definitionadditionalflowsonphi2}, we get
\begin{eqnarray*}
(\partial_{t^*_{m,l}}\L)&=& (\partial_{t^*_{m,l}}S)\ \La S^{-1}
+ S\ \La\ (\partial_{t_{m,l}}S^{-1})\\
&=&-((M-\bar M)^m\L^l)_{-} S\ \La\ S^{-1}- S\ \La
S^{-1}\ (\partial_{t^*_{m,l}}S)
\ S^{-1}\\
&=&-((M-\bar M)^m\L^l)_{-} \L+ \L ((M-\bar M)^m\L^l)_{-}\\
&=&-[((M-\bar M)^m\L^l)_{-},\L].
\end{eqnarray*}
Similarly, we perform the derivative on $\L$ dressed by $\bar S$ and
use the additional flow about $\bar S$ in \eqref{definitionadditionalflowsonphi2} to get the following
\begin{eqnarray*}
(\partial_{t^*_{m,l}}\L)&=& (\partial_{t^*_{m,l}}\bar S)\ \La \bar S^{-1}
+ \bar S\ \La\ (\partial_{t_{m,l}}\bar S^{-1})\\
&=&((M-\bar M)^m\L^l)_{+} \bar S\ \La^{-1}\ \bar S^{-1}- \bar S\ \La
\bar S^{-1}\ (\partial_{t^*_{m,l}}\bar S)
\ \bar S^{-1}\\
&=&((M-\bar M)^m\L^l)_{+} \L- \L ((M-\bar M)^m\L^l)_{+}\\
&=&[((M-\bar M)^m\L^l)_{+},\L].
\end{eqnarray*}
Because
\begin{eqnarray}\label{QTHadditionalflow111.}
[M-\bar M,\L]=0,
\end{eqnarray}
therefore
\begin{eqnarray}\label{QTHadditionalflow1111}
\dfrac{\partial \L}{\partial
{t^*_{m,l}}}=[-\left((M-\bar M)^m\L^l\right)_{-},
\L]=[\left((M-\bar M)^m\L^l\right)_{+}, \L],
\end{eqnarray}
which gives the compatibility of additional flow of QTH with reduction condition \eqref{two dressing}.
\end{proof}

Similarly, we can take derivatives on the dressing structure of  $M$ and  $\bar M$ to get the following proposition.
\begin{proposition}\label{add flow}
The additional derivatives  act on  $M$, $\bar M$ as
\begin{eqnarray}
\label{QTHadditionalflow11'}
\dfrac{\partial
M}{\partial{t^*_{m,l}}}&=&[-\left((M-\bar M)^m\L^l\right)_{-}, M],
\\
\label{QTHadditionalflow12}
\dfrac{\partial
\bar M}{\partial{t^*_{m,l}}}&=&[\left((M-\bar M)^m\L^l\right)_{+}, \bar M].
\end{eqnarray}
\end{proposition}
\begin{proof} By performing the derivative on  $M$ given in (\ref{Moperator}), there exists
a similar derivative as $\partial_{t^*_{m,l}}\L$, i.e.,
\begin{eqnarray*}
(\partial_{t^*_{m,l}}M)&\!\!\!=\!\!\!&(\partial_{t^*_{m,l}}S)\ \Gamma  S^{-1}
+ S\ \Gamma \ (\partial_{t^*_{m,l}}S^{-1})\\
&\!\!\!=\!\!\!&-((M-\bar M)^m\L^l)_{-} S\ \Gamma \ S^{-1}- S\ \Gamma
S^{-1}\ (\partial_{t^*_{m,l}}S)
\ S^{-1}\\
&\!\!\!=\!\!\!&-((M-\bar M)^m\L^l)_{-} M+ M
((M-\bar M)^m\L^l)_{-}\\
&=&-[((M-\bar M)^m\L^l)_{-}, M].
\end{eqnarray*}
Here the fact that $\Gamma $ does not depend on the additional
variables $t^*_{m,l}$ has been used. Other identities can also
be obtained in a similar way.
\end{proof}

By the two propositions above, the following theorem can be proved.
\begin{theorem}\label{symmetry}
The additional flows $\partial_{t^*_{m,l}}$ commute
with the $q$-Toda hierarchy flows $\partial_{t_{n}}$, i.e.,
\begin{eqnarray}
[\partial_{t^*_{m,l}}, \partial_{t_{n}}]\Phi=0,
\end{eqnarray}
where $\Phi$ can be $S$, $\bar S$ or $\L$,  and
 $
\partial_{t^*_{m,l}}=\frac{\partial}{\partial{t^*_{m,l}}},
\partial_{t_{n}}=\frac{\partial}{\partial{t_{n}}}$.

\end{theorem}
\begin{proof} According to the definition,
\begin{eqnarray*}
[\partial_{t^*_{m,l}},\partial_{t_{n}}]S=\partial_{t^*_{m,l}}
(\partial_{t_{n}}S)-
\partial_{t_{n}} (\partial_{t^*_{m,l}}S),
\end{eqnarray*}
and using the actions of the additional flows and the $q$-Toda
flows on $S$, we have
\begin{eqnarray*}
[\partial_{t^*_{m,l}},\partial_{t_{n}}]S
&=& -\partial_{t^*_{m,l}}\left((B_{n})_{-}S\right)+
\partial_{t_{n}} \left(((M-\bar M)^m\L^l)^m_{-}S \right)\\
&=& -(\partial_{t^*_{m,l}}B_{n} )_{-}S-
(B_{n})_{-}(\partial_{t^*_{m,l}}S)\\&&+
[\partial_{t_{n}} ((M-\bar M)^m\L^l)]_{-}S +
((M-\bar M)^m\L^l)_{-}(\partial_{t_{n}}S).
\end{eqnarray*}
Using \eqref{definitionadditionalflowsonphi2} and Proposition \ref{flowsofM}, it
equals
\begin{eqnarray*}
[\partial_{t^*_{m,l}},\partial_{t_{n}}]S
&=&[\left((M-\bar M)^m\L^l\right)_{-}, B_{n}]_{-}S+
(B_{n})_{-}\left((M-\bar M)^m\L^l\right)_{-}S\\
&&+[(B_{n})_{+},(M-\bar M)^m\L^l]_{-}S-((M-\bar M)^m\L^l)_{-}(B_{n})_{-}S\\
&=&[((M-\bar M)^m\L^l)_{-}, B_{n}]_{-}S- [(M-\bar M)^m\L^l,
(B_{n})_{+}]_{-}S\\&&+
[(B_{n})_{-},((M-\bar M)^m\L^l)_{-}]S\\
&=&0.
\end{eqnarray*}
In the proof above, $[(B_{n})_{+},
((M-\bar M)^m\L^l)]_{-}= [(B_{n})_{+}, ((M-\bar M)^m\L^l)_{-}]_{-}$ has
been used. The action on $\L$ in the theorem can be proved in similar ways.
\end{proof}
The commutative property in Theorem \ref{symmetry} means that
additional flows are symmetries of the QTH.
Since they are symmetries, it is natural to consider the algebraic
structures among these additional symmetries. So we obtain the following important
theorem.
\begin{theorem}\label{WinfiniteCalgebra}
The additional flows  $\partial_{t^*_{m,l}}$ form a Block type Lie algebra with the
following relation
 \begin{eqnarray}\label{algebra relation}
[\partial_{t^*_{m,l}},\partial_{t^*_{n,k}}]= (km-n l)\d^*_{m+n-1,k+l-1},
\end{eqnarray}
which holds in the sense of acting on  $S$, $\bar S$ or $\L$ and  $m,n,l,k\geq 0.$
\end{theorem}
\begin{proof}
 By using
 (\ref{definitionadditionalflowsonphi2}), we get
\begin{eqnarray*}
[\partial_{t^*_{m,l}},\partial_{t^*_{n,k}}]S&=&
\partial_{t^*_{m,l}}(\partial_{t^*_{n,k}}S)-
\partial_{t^*_{n,k}}(\partial_{t^*_{m,l}}S)\\
&=&-\partial_{t^*_{m,l}}\left(((M-\bar M)^n\L^k)_{-}S\right)
+\partial_{t^*_{n,k}}\left(((M-\bar M)^m\L^l)_{-}S\right)\\
&=&-(\partial_{t^*_{m,l}}
(M-\bar M)^n\L^k)_{-}S-((M-\bar M)^n\L^k)_{-}(\partial_{t^*_{m,l}} S)\\
&&+ (\partial_{t^*_{n,k}} (M-\bar M)^m\L^l)_{-}S+
((M-\bar M)^m\L^l)_{-}(\partial_{t^*_{n,k}} S).
\end{eqnarray*}
We further get
 \begin{eqnarray*}&&
[\partial_{t^*_{m,l}},\partial_{t^*_{n,k}}]S\\
&=&-\Big[\sum_{p=0}^{n-1}
(M-\bar M)^p(\partial_{t^*_{m,l}}(M-\bar M))(M-\bar M)^{n-p-1}\L^k
+(M-\bar M)^n(\partial_{t^*_{m,l}}\L^k)\Big]_{-}S\\&&-((M-\bar M)^n\L^k)_{-}(\partial_{t^*_{m,l}} S)\\
&&+\Big[\sum_{p=0}^{m-1}
(M-\bar M)^p(\partial_{t^*_{n,k}}(M-\bar M))(M-\bar M)^{m-p-1}\L^l
+(M-\bar M)^m(\partial_{t^*_{n,k}}\L^l)\Big]_{-}S\\&&+
((M-\bar M)^m\L^l)_{-}(\partial_{t^*_{n,k}} S)\\
&=&[(nl-km)(M-\bar M)^{m+n-1}\L^{k+l-1}]_-S\\
&=&(km-nl)\d^*_{m+n-1,k+l-1}S.
\end{eqnarray*}
Similarly  the same results on $\bar S$ and $\L$ are as follows
 \begin{eqnarray*}
[\partial_{t^*_{m,l}},\partial_{t^*_{n,k}}]\bar S
&=&((km-nl)(M-\bar M)^{m+n-1}\L^{k+l-1})_+\bar S\\
&=&(km-nl)\d^*_{m+n-1,k+l-1}\bar S,
\\[6pt]
{}[\partial_{t^*_{m,l}},\partial_{t^*_{n,k}}]\L&=&
\partial_{t^*_{m,l}}(\partial_{t^*_{n,k}}\L)-
\partial_{t^*_{n,k}}(\partial_{t^*_{m,l}}\L)\\
&=&[((nl-km)(M-\bar M)^{m+n-1}\L^{k+l-1})_-, \L]\\
&=&(km-nl)\d^*_{m+n-1,k+l-1}\L.
\end{eqnarray*}
\end{proof}
Denote $D_{m,l}=\d_{t^*_{m+1,l+1}}$, and let Block algebra be the span of all $D_{m,l},\,m,l\ge-1$.
Then by \eqref{algebra relation}, Block algebra is a Lie algebra with relations
\begin{eqnarray}
[D_{m,l},D_{n,k}]=((m+1)(k+1)-(l+1)(n+1))D_{m+n,l+k},\mbox{ \ for \ } m,n,l,k\geq -1.
\end{eqnarray}
Thus Block algebra is in fact a Block type Lie algebra which is generated by the set
\begin{eqnarray}\label{Gen-B}
B=\{ D_{-1,0},
D_{0,-1}, D_{0,0}, D_{1,0}, D_{0,1}\}=\{ \d^*_{0,1}, \d^*_{1,0}, \d^*_{1,1}, \d^*_{2,1},
\d^*_{1,2}\}.
\end{eqnarray}

\begin{theorem}
The Block flows of the $q$-Toda hierarchy  are Hamiltonian systems
in the form
\begin{equation}
\frac{\d u}{\d t^*_{m,l}}=\{u,H^*_{m,l}\}_1,\ \frac{\d v}{\d t^*_{m,l}}=\{v,H^*_{m,l}\}_1,\ \ m,l\geq 0.
\end{equation}
They satisfy the following bi-Hamiltonian recursion relation
\[ \frac{\d }{\d t^*_{m,l}}=\{\cdot,H^*_{m,l-1}\}_2=n
\{\cdot,H^*_{m,l}\}_1.
\]
Here the Hamiltonians (depending on $t_{n}$) with respect to $t^*_{m,l}$ have the form
\begin{equation}
H^*_{m,l}=\int h^*_{m,l}(u,v; u_x,v_x; \dots;t_{n}; \epsilon) dx,\quad  \ n\ge 0,
\end{equation}
with the Hamiltonian densities  $ h^*_{m,l}(u,v; u_x,v_x; \dots;t_{n}; \epsilon)$ given by
\[
 h^*_{m,l}&=&\res (M-\bar M)^m\L^l.
\]

\end{theorem}
\begin{proof}
The proof is similar as the proof for original Toda flows.
\end{proof}

\section{Extended $q$-Toda hierarchy}
 To define the extended flows, we define the following logarithm
\begin{align}
\log_+\L&=W\circ\epsilon x\partial\circ W^{-1}=S\circ\epsilon x\partial\circ S^{-1},\\
\log_-\L&=-\bar W\circ\epsilon x \partial\circ \bar W^{-1}=-\bar S\circ\epsilon x\partial\circ \bar S^{-1},
\end{align}
where $\d$ is the derivative about the spatial variable $x$.

Combining these above logarithmic operators together can derive following important logarithm
\begin{align}
\label{Log} \log \L:&=\frac12(\log_+\L+\log_-\L)=\frac12(S\circ\epsilon x\partial\circ S^{-1}-\bar S\circ\epsilon x\partial\circ \bar S^{-1}):=\sum_{i=-\infty}^{+\infty}W_i\Lambda_q^i\in G,
\end{align}
which will generate a series of flow equations which contain the spatial flow in later defined Lax equations.
Let us first introduce some convenient notations.
\begin{definition}The  operators $B_{j},D_{j}$ are defined as follows
\begin{align}\label{satoS}
\begin{aligned}
B_{j}&:=\frac{\L^{j+1}}{(j+1)!},\ \
D_{j}:=\frac{2\L^j}{j!}(\log \L-c_j),\ \  c_j=\sum_{i=1}^j\frac 1i,\ j\geq 0.
\end{aligned}
\end{align}
\end{definition}

Now we give the definition of the extended $q$-Toda hierarchy(EQTH).
\begin{definition}The extended $q$-Toda hierarchy is a hierarchy in which the dressing operators $S,\bar S$ satisfy following Sato equations
\begin{align}
\label{satoSt} \epsilon\partial_{t_{j}}S&=-(B_{j})_-S,& \epsilon\partial_{t_{j}}\bar S&=(B_{j})_+\bar S,  \\
\label{satoSs}\epsilon\partial_{ s_{j}}S&=-(D_{j})_- S,& \epsilon\partial_{s_{j}}\bar S&=(D_{j})_+\bar S.\end{align}
\end{definition}
Then one can easily get the following proposition about $W,\bar W.$

\begin{proposition}The dressing operators $W,\bar W$ are subject to following Sato equations
\begin{align}
\label{Wjk} \epsilon\partial_{t_{j}}W&=(B_{j})_+ W,& \epsilon\partial_{t_{j}}\bar W&=(B_{j})_+\bar W,  \\
\epsilon\partial_{s_{j}}W&=(\frac{\L^j}{j!}(\log_+ \L-c_j) -(D_{j})_-) W,& \epsilon\partial_{s_{j}}\bar W&=(-\frac{\L^j}{j!}(\log_- \L-c_j)+(D_{j})_+)\bar W.  \end{align}
\end{proposition}

 From the previous proposition we derive the following  Lax equations for the Lax operators.
\begin{proposition}\label{Lax}
 The  Lax equations of the EQTH are as follows
   \begin{align}
\label{laxtjk}
  \epsilon\partial_{t_{j}} \L&= [(B_{j})_+,\L],&
  \epsilon\partial_{s_{j}} \L&= [(D_{j})_+,\L],\
  \epsilon\partial_{ t_{j}} \log \L= [(B_{j})_+ ,\log \L],&
  \end{align}
   \begin{align}\epsilon(\log \L)_{ s_{j}}=[ -(D_{j})_-,\log_+ \L ]+
[(D_{j})_+ ,\log_- \L ].
\end{align}
\end{proposition}

To see this kind of hierarchy more clearly, the Hirota quadratic equations of the EQTH  will be given in next subsection.

\section{Generalized  vertex operators and  Hirota quadratic equations}

Introduce the following sequences:
\[t-[\lambda] &:=& (t_{j}-
  \epsilon(j-1)!\lambda^j, 0\leq j\leq \infty).
\]
A scalar function   depending only on the dynamical variables $t,s$ and
$\epsilon$ is called the  {\em \bf tau-function of the EQTH} if it
provides symbols related to wave operators as following,

\begin{eqnarray}\label{Mpltaukk}\bS: &=&\frac{ \tau
(e^{s_{0}-\frac{\epsilon}{2}}x, t_{j}-\frac{\epsilon(j-1)!}{\lambda^j},s;\epsilon) }
     {\tau (e^{s_{0}-\frac{\epsilon}{2}}x,t,s;\epsilon)},\\
     \label{Mpl-1taukk}\bS^{-1}: &=&\frac{ \tau
(e^{s_{0}+\frac{\epsilon}{2}}x, t_{j}+\frac{\epsilon(j-1)!}{\lambda^j},s;\epsilon) }
     {\tau (e^{s_{0}+\frac{\epsilon}{2}}x,t,s;\epsilon)},\\ \label{Mprtaukk}
\bar \bS:&= &\frac{ \tau
(e^{s_{0}+\frac{\epsilon}{2}}x,t_{j}+\epsilon(j-1)!\lambda^j,s;\epsilon)}
     {\tau(e^{s_{0}-\frac{\epsilon}{2}}x,t,s;\epsilon)},\\
     \bar \bS^{-1}:&= &\frac{\tau
(e^{s_{0}-\frac{\epsilon}{2}}x,t_{j}-\epsilon(j-1)!\lambda^j,s;\epsilon)}
     {\tau(e^{s_{0}+\frac{\epsilon}{2}}x,t,s;\epsilon)}.
     \end{eqnarray}
The proof of the existence of the tau function of the EQTH is a also standard, one can refer the similar proof in \cite{M,ourJMP}.

{\bf Remark:} We need to note that the tau function of the EQTH is unique up to a multiplication of an arbitrary function depending on extended variables $s_j,\ j>0$ for a pair of given wave functions.

In this section we continue to  discuss on the fundamental properties
of the tau function of the EQTH, i.e., the Hirota quadratic equations of the EQTH. So we
introduce the following vertex operators
\begin{eqnarray*}
\Gamma^{\pm a} :&=&\exp\left(\pm \frac{1}{\epsilon}
(\sum_{j=0}^\infty t_{j}\frac{\lambda^{j+1}}{(j+1)!}+ s_{j}\frac{\lambda^{j}}{j!}( \log \lambda-c_j))\right)\times\exp\left({\mp
\frac{\epsilon}{2}\partial_{s_0} \mp [\lambda^{-1}]_\d  }\right),\\
\Gamma^{\pm b} :&=&\exp\left(\pm  \frac{1}{\epsilon}
(\sum_{j=0}^\infty t_{j}\frac{\lambda^{-j-1}}{(j+1)!}- s_{j}\frac{\lambda^{-j}}{j!}( \log \lambda-c_j))\right)\times\exp\left({\mp
\frac{\epsilon}{2}\partial_{s_0} \mp [\lambda]_{\d}  }\right),
\end{eqnarray*}

where\ \
\begin{eqnarray*}
 [\lambda]_\d  :&=&
\epsilon\sum_{j=0}^\infty j!\lambda^{j+1}\partial_{t_{j}}.
\end{eqnarray*}

Because of the logarithm $\log \lambda$, the vertex
operators  $\Gamma^{\pm a} \otimes \Gamma^{\mp a}$ and
$\Gamma^{\pm b} \otimes \Gamma^{\mp b }$  are multi-valued. There
are monodromy factors $M^a$ and $M^b$ respectively as following
among different branches around $\lambda=\infty$
\begin{equation} M^{a}= \exp \left\{ \pm \frac{2\pi i}{\epsilon}
\sum_{j\geq 0}\frac{\lambda^{j}}{j!} ( s_{j} \otimes 1 - 1\otimes s_{j})
\right\},\end{equation}
\begin{equation}
 M^{b}= \exp \left\{ \pm \frac{2\pi i}{\epsilon}
\sum_{j\geq 0}\frac{\lambda^{-j}}{j!} ( s_{j} \otimes 1 - 1\otimes s_{j})
\right\}.
\end{equation}
In order to offset the complication, we need to generalize the
concept of vertex operators which leads it to be not scalar-valued
any more but take values in a differential operator algebra in $\C$. So we introduce the following vertex operators
\begin{equation}\Gamma^{\delta}_{a} = \exp\( -\sum_{j>0}\frac{j!\lambda^{j+1}}{\epsilon
}(\epsilon x\d_x)s_{j}\) \exp(\log x\  \partial_{s_0}),\end{equation}
\begin{equation}\Gamma^{\delta}_{b} = \exp\( -\sum_{j>0}\frac{j!\lambda^{-(j+1)}}{\epsilon
}(\epsilon x\d_x)s_{j}\) \exp(\log x\  \partial_{s_0}),\end{equation}
\begin{equation}\Gamma^{\delta
\#}_{a} =\exp(\log x\  \partial_{s_0}) \exp\( \sum_{j>0}\frac{j!\lambda^{j+1}}{\epsilon
}(\epsilon x\d_x)s_{j}\) ,\end{equation}
\begin{equation}\Gamma^{\delta\#}_{b} =  \exp(\log x\  \partial_{s_0})\exp\( \sum_{j>0}\frac{j!\lambda^{-(j+1)}}{\epsilon
}(\epsilon x\d_x)s_{j}\).\end{equation}
 Then \begin{equation}
 \label{double delta a} \Gamma^{\delta
\#}_{a}\otimes \Gamma^{\delta}_{a} = \exp(\log x\  \partial_{s_0})\exp\(
\sum_{j>0}\frac{j!\lambda^{j+1}}{\epsilon
}(\epsilon x\d_x)(s_{j}-s'_{j}) \) \exp(\log x\  \partial_{s'_{0}}),
\end{equation}
\begin{equation}
\label{double delta b} \ \ \ \Gamma^{\delta \#}_{b}\otimes
\Gamma^{\delta}_{b} = \exp(\log x\  \partial_{s_0})\exp\(
\sum_{j>0}\frac{j!\lambda^{-(j+1)}}{\epsilon
}(\epsilon x\d_x)(s_{j}-s'_{j}) \) \exp(\log x\  \partial_{s'_{0}}).
\end{equation}

After computation we get
\begin{eqnarray*} && \(\Gamma^{\delta \#}_{a}\otimes \Gamma^\delta_{a} \) M^{a} =
\exp \left\{ \pm \frac{2\pi i}{\epsilon} \sum_{j> 0} \frac{\lambda^{j}}{j!} ( s_{j}-s'_{j})
\right\}\\
&& \exp\(  \pm \frac{2\pi i}{\epsilon} ((s_0+\log x) -(s'_{0}+\log x+ \sum_{j> 0} \frac{\lambda^{j}}{j!} ( s_{j}-s'_{j})) \) \(\Gamma^{\delta \#}_{a}\otimes \Gamma^\delta _{a}\)
\\&=& \exp\({\pm \frac{2\pi i}{\epsilon}(s_0-s'_{0})}\)
\(\Gamma^{\delta \#}_{a}\otimes \Gamma^\delta _{a}\),
\end{eqnarray*}
\begin{eqnarray*}
&& \(\Gamma^{\delta \#}_{b}\otimes \Gamma^\delta_{ b} \) M^{b} =
\exp \left\{ \pm \frac{2\pi i}{\epsilon} \sum_{j> 0} \frac{\lambda^{-j}}{j!} ( s_{j}-s'_{j})
\right\}\\
&& \exp\(  \pm \frac{2\pi i}{\epsilon} ( (s_0+\log x) -(
s'_{0}+\log x+ \sum_{j> 0} \frac{\lambda^{-j}}{j!} ( s_{j}-s'_{j})) \)\(\Gamma^{\delta \#}_{b}\otimes \Gamma^\delta_{b} \)
\\&=& \exp\({\pm \frac{2\pi i}{\epsilon}(s_0-s'_{0})}\)
\(\Gamma^{\delta \#}_{b}\otimes \Gamma^\delta_{b} \).
\end{eqnarray*}
Thus when $s_0-s'_{0} \in \Z\epsilon $, $\(\Gamma^{\delta
\#}_{a}\otimes \Gamma^{\delta}_{a}\) \( \Gamma^{a}\otimes
\Gamma^{-a}\) \mbox{and}\(\Gamma^{\delta \#}_{b}\otimes
\Gamma^{\delta}_{b }\)\(\Gamma^{-b}\otimes\Gamma^{b}\)$ are all
single-valued near $\lambda=\infty$.

 Now we should note that the above vertex operators  take
 value in  differential operator algebra $\C[\d,x,t,s,\epsilon]:=\{f(x,t,\epsilon)|f(x,t,s,\epsilon)=\sum_{i\geq 0}c_{i}(x,t,s,\epsilon)\d^i\}$.
Then we can get the following important theorem similar as \cite{M,ourJMP}.
\begin{theorem}\label{t11}
The invertible   $\tau(t,s,\epsilon)$  is a tau-function of the EQTH if and only if it
satisfies the  following Hirota quadratic equations  of the EQTH.
\begin{equation} \label{HBE'} \res_{{\rm{\lambda}}}
 \lambda^{r-1}\(\Gamma^{\delta
\#}_{a}\otimes \Gamma^{\delta}_{a}\) \( \Gamma^{a}\otimes
\Gamma^{-a}\right)(\tau
\otimes \tau ) =\res_{{\rm{\lambda}}}\lambda^{-r-1}\(\Gamma^{\delta \#}_{b}\otimes
\Gamma^{\delta}_{b }\)\(\Gamma^{-b}\otimes\Gamma^{b} \) (\tau
\otimes \tau )
\end{equation}
computed at $s_0-s'_{0}=l\epsilon$
 for each  $l\in \Z$, $r\in \N$.
\end{theorem}

\section{Bi-Hamiltonian structure of the EQTH}
As another important part of the Sato theory, the bi-Hamiltonian structure of the EQTH will be constructed in the next section similar as \cite{CDZ}.
\begin{theorem}
The flows of the EQTH  are Hamiltonian systems
of the form
\[
\frac{\d u_i}{\d t_{k,j}}&=&\{u_i,H_{k,j}\}_1, \  \frac{\d v_i}{\d t_{k,j}}=\{v_i,H_{k,j}\}_1,
\quad k=0,1;\ j\ge 0,
\label{td-ham}
\]
with $ t_{0,j}=t_{j},t_{1,j}=s_{j}.$
They satisfy the following bi-Hamiltonian recursion relation
\[\notag
\{\cdot,H_{1,n-1}\}_2&=&n
\{\cdot,H_{1,n}\}_1+2\{\cdot,H_{0,n-1}\}_1,\ \{\cdot,H_{0,n-1}\}_2=(n+1)
\{\cdot,H_{0,n}\}_1.
\]
Here the Hamiltonians have the form
\begin{equation}
H_{k,j}=\int h_{k,j}(u,v; u_x,v_x; \dots; \epsilon) dx,\quad k=0,1; \ j\ge 0,
\end{equation}
with
\[
 h_{0,j}&=&\frac1{(j+1)!} Res \, \L^{j+1},\
 h_{1,j}=\frac2{j!}\, Res\left[ \L^{j}
(\log \L-c_{j})\right].
\]

\end{theorem}

\begin{proof}
For the $q$-Toda hierarchy, the proof was already given in  the Theorem \ref{QTHbiha}.

Here we will prove that the flows $\frac{\d}{\d t_{1,n}}$ are also
Hamiltonian systems with respect to the first Poisson bracket.
Like in \cite{CDZ}, the following identity has been proved
\begin{equation}\label{dlgl-2}  Res\left[\L^n d (S\epsilon x\d_x
S^{-1})\right] \sim  Res \L^{n-1} d \L,
\end{equation}
which show the validity of the following equivalence relation:
\begin{equation}\label{dlglg}
 Res\left(\L^n\, d \log_+ \L\right) \sim  Res\left(\L^{n-1} d \L\right).
\end{equation}
Here the equivalent relation $\sim$ is up to a $x$-derivative of
another 1-form.

In a similar way as eq.\eqref{dlgl-2}, we obtain the following equivalence relation
\begin{equation}\label{dlgl-3}
 Res\left[\L^n d (\bar S\epsilon x\d_x \bar S^{-1})\right]\sim -\rm  Res \L^{n-1} d \L.
\end{equation}
i.e.
\begin{equation}\label{dlglg'}
 Res\left(\L^n\, d \log_- \L\right) \sim \rm  Res\left(\L^{n-1} d \L\right).
\end{equation}
Combining \eqref{dlglg} with \eqref{dlglg'} together can lead to
\begin{equation}\label{dlgl2}
 Res\left(\L^n\, d \log \L\right) \sim \rm  Res\left(\L^{n-1} d \L\right).
\end{equation}

Then from
\begin{equation}
  \label{edef3}
\frac{\partial \L}{\partial t_{k, n}} = [ (B_{k,n})_+ ,\L ]= [ -(B_{k,n})_- ,\L ], \ \ B_{0,n}=B_n,B_{1,n}=D_n,
\end{equation}
and supposing
\[
B_{1,n}=\sum_{k} a_{1,n+1;k}\, \Lambda^k,
\]
we can derive equation
\[\epsilon\frac{\partial u}{\partial t_{1, n}}&=&a_{1,n+1;1}(qx)-a_{1,n+1;1}(x)\in \C,\ \\
\epsilon\frac{\partial v}{\partial t_{1, n}}&=&a_{1,n+1;0}(q^{-1}x) e^{v(x)}-a_{1,n+1;0}(x) e^{v(qx)}\in \C.
\]

The equivalence relation (\ref{dlgl2}) now readily follows from the above two equations.
By using (\ref{dlglg}) we obtain
\begin{eqnarray}
&&d  h_{1,n}=\frac2{n!}\,d\, Res\left[\L^{n}
\left(\log \L-c_{n}\right) \right]
\notag\\
&& \sim \frac2{(n-1)!}\, Res\left[\L^{n-1}
\left(\log \L-c_{n}\right) d \L\right]+ \frac2{n!}\, Res\left[\L^{n-1} d \L\right]\notag\\
&&=\frac2{(n-1)!}\, Res\left[\L^{n-1} \left(\log \L-c_{n-1}\right) d \L\right]\\
&&= Res\left[a_{1,n;0}(x)du+a_{1,n;1}(q^{-1}x) e^{v(x)}dv\right].
\end{eqnarray}
It yields the following identities
\begin{equation}\label{dH1-u12}
\frac{\delta H_{1,n}}{\delta u}=a_{1,n;0}(x),\quad \frac{\delta H_{1,n}}
{\delta v}=a_{1,n;1}(q^{-1}x) e^{v(x)}.
\end{equation}
This agree with Lax equation

\[\notag
\frac{\d u}{\d t_{1,n}}&=&\{u,H_{1,n}\}_1={1\over \epsilon} \left[
e^{\epsilon\,x\d_x}-1\right]\frac{\delta H_{1,n}}
{\delta v}={1\over \epsilon}(a_{1,n+1;1}(qx)-a_{1,n+1;1}(x)),\\ \notag
 \  \frac{\d v}{\d t_{1,n}}&=&\{v,H_{1,n}\}_1=\frac{1}{\epsilon} \left[1-e^{\epsilon\,x\d_x}
\right]\frac{\delta H_{1,n}}
{\delta u}=\frac{1}{\epsilon} \left[a_{1,n+1;0}(q^{-1}x) e^{v(x)}-a_{1,n+1;0}(x) e^{v(qx)}\right].
\]

 From the above identities we see that
the flows $\frac{\d}{\d t_{1,n}}$ are Hamiltonian systems
of the first bi-Hamiltonian structure.
For the case of $k=1$ the recursion relation
follows from the following trivial identities
\begin{eqnarray}
&&n\, \frac{2}{n!} \L^{n} \left(\log_{\pm} \L-c_{n}\right)=\L\,
\frac{2}{(n-1)!}
\L^{n-1} \left(\log_{\pm} \L-c_{n-1}\right)-2\,\frac1{n!} \L^n\notag\\
&&=\frac{2}{(n-1)!} \L^{n-1} \left(\log_{\pm} \L-c_{n-1}\right)\,
\L-2\,\frac1{n!} \L^n.\notag
\end{eqnarray}
Then we get, for $\beta=1,$
\begin{eqnarray}
&&n a_{1,n+1;1}(x)=a_{1,n;0}(qx)+ua_{1,n;1}(x)+e^va_{1,n;2}(q^{-1}x)-2a_{0,n+1;1}(x)\notag\\
&&=a_{1,n;0}(x)+u(qx)a_{1,n;1}(x)+e^{v(q^2x)}a_{1,n;2}(x)-2a_{0,n+1;1}(x).\notag
\end{eqnarray}
This further leads to

\begin{eqnarray}
&&\{u,H_{1,n-1}\}_2=\{\left[\Lambda e^{v(x)}-e^{v(x)} \Lambda^{-1}\right] a_{1,n;0}(x)+
u(x) \left[\Lambda-1\right] a_{1,n;1}(q^{-1}x) e^{v(x)}\}\notag\\ \notag
&&
=n\left[a_{1,n+1;1}(x) e^{v(qx)}-a_{1,n+1;1}(q^{-1}x) e^{v(x)}\right]+2\left[a_{0,n+1;0}(x) e^{v(qx)}-a_{0,n+1;0}(q^{-1}x) e^{v(x)}\right].\label{pre-recur}
\end{eqnarray}
This is exactly the recursion relation on flows for $u$. The similar recursion flow on $v$ can be similarly derived.
Theorem is proved till now.

\end{proof}

Similarly as \cite{CDZ}, the tau symmetry of the EQTH can be proved in the  following theorem.
\begin{theorem}\label{tausymmetry}
The EQTH has the following tau-symmetry property:
\begin{equation}
\frac{\d h_{\alpha,m}}{\d t_{\beta, n}}=\frac{\d
h_{\beta, n}}{\d t_{\alpha,m}},\quad \alpha,\beta=0,1,\ m,n\ge 0.
\end{equation}
\end{theorem}
\begin{proof} Let us prove the theorem for the case when $\alpha=1, \beta=0$,
other cases are proved in a similar way
\[
&&\frac{\d h_{1,m}}{\d t_{0,n}} =\frac2{m!\,(n+1)!}\, Res[-(\L^{n+1})_-, \L^m (\log\L-c_m)]\notag\\
&&=\frac2{m!\,(n+1)!}\, Res[(\L^m (\log\L-c_m))_+,(\L^{n+1})_-]\notag\\
&& =\frac2{m!\,(n+1)!}\, Res[(\L^m (\log\L
-c_m))_+,\L^{n+1}]=\frac{\d h_{0,n}}{\d t_{1,m}}.
\]
The theorem is proved.
\end{proof}

 This property justifies the following alternative definition of another kind of
tau function for the EQTH.

\begin{definition} The  $tau$ function $\bar \tau$ of the EQTH can be defined by
the following expressions in terms of the densities of the Hamiltonians:
\begin{equation}
h_{\beta,n}=\epsilon (\Lambda-1)\frac{\d\log \bar \tau}{\d t_{\beta,n}},
\quad \beta=0,1;\ n\ge 0,
\end{equation}
with $ t_{0,j}=t_{j},t_{1,j}=s_{j}.$
\end{definition}

With above two different definitions tau functions of this hierarchy, some mysterious connections between these two kinds of tau functions become an open question. One is from Sato theory without fixing extended variables and another is from the Hamiltonian tau symmetry. While considering the constraint of $\bar \tau$ down to a sub-manifold only depending on non-extended coordinates,  $\bar \tau$ and $\tau$  should be the same.

\section{Darboux transformation of the EQTH}

In this section, we will consider the Darboux transformation of the EQTH on Lax operator
 \[\L=\Lambda_q+u+v\Lambda_q^{-1},\]
 i.e.
  \[\label{1darbouxL}\L^{[1]}=\Lambda_q+u^{[1]}+v^{[1]}\Lambda_q^{-1}=W\L W^{-1},\]
where $W$ is the Darboux transformation operator.
That means after Darboux transformation, the spectral problem about the wave function $\phi$

\[\L\phi=\Lambda_q\phi+u\phi+v\Lambda_q^{-1}\phi=\lambda \phi,\]
will become

\[\L^{[1]}\phi^{[1]}=\lambda\phi^{[1]}.\]

To keep the Lax pair of the EQTH invariant , i.e.
   \begin{align}
\label{laxtjk}
  \epsilon \partial_{ t_{j}} \L^{[1]}&= [(B_{j}^{[1]})_+,\L^{[1]}],
 \epsilon \partial_{ s_{j}} \L^{[1]}= [(D_{j}^{[1]})_+,\L^{[1]}],\ \  B_{j}^{[1]}:=B_{j}(\L^{[1]}), D_{j}^{[1]}:=D_{j}(\L^{[1]}),\\
   \epsilon \partial_{  t_{j}} \log \L^{[1]}&= [(B_{j}^{[1]})_+ ,\log \L^{[1]}],\ \  \epsilon \partial_{  s_{j}}\log \L^{[1]}=[ -(D_{j}^{[1]})_-,\log_+ \L^{[1]} ]+
[(D_{j}^{[1]})_+ ,\log_- \L^{[1]} ],
\end{align} the dressing operator $W$ should satisfy the following dressing equation
\[ \epsilon \partial_{  t_{j}}W&=&-W(B_{j})_++(WB_{j}W^{-1})_+W,\ \  j\geq 0\\
 \epsilon \partial_{  s_{j}}W&=&-W(D_{j})_++(WD_{j}W^{-1})_+W,\ \  j\geq 0.\]
where $W_{t_{j}}$ means the derivative of $W$ by $t_{j}.$
To give the Darboux transformation, we need the following lemma.
\begin{lemma}\label{lema}
The operator $B:=\sum_{n=0}^{\infty}b_n\Lambda_q^n$ is a non-negative difference operator,  $C:=\sum_{n=1}^{\infty}c_n\Lambda_q^{-n}$ is a negative difference operator and $f,g$ (short for $f(x),g(x)$) are two functions of the spatial parameter $x$, following identities hold
\begin{equation}\label{Bneg}
(Bf \frac{\Lambda_q^{-1}}{1-\Lambda_q^{-1}} g)_-=B(f) \frac{\Lambda_q^{-1}}{1-\Lambda_q^{-1}} g,\ \ \ (f \frac{\Lambda_q^{-1}}{1-\Lambda_q^{-1}} gB)_-=f \frac{\Lambda_q^{-1}}{1-\Lambda_q^{-1}}B^*(g),
\end{equation}
\begin{equation}\label{Cpos}
(Cf \frac{1}{1-\Lambda_q} g)_+=C(f)\frac{1}{1-\Lambda_q} g,\ \ \ (f \frac{1}{1-\Lambda_q} gC)_+=f \frac{1}{1-\Lambda_q}C^*(g).
\end{equation}
\end{lemma}
\begin{proof}
Here we only give the proof of the eq.\eqref{Bneg} by direct calculation
\[\notag
(Bf \frac{\Lambda_q^{-1}}{1-\Lambda_q^{-1}} g)_-&=&\sum_{m=0}^{\infty}b_m(f(q^mx)\La^m \frac{\Lambda_q^{-1}}{1-\Lambda_q^{-1}} g)_-\\ \notag
&=&\sum_{m=0}^{\infty}b_mf(q^mx)(\frac{\Lambda_q^{m-1}}{1-\Lambda_q^{-1}})_- g\\ \notag
&=&\sum_{m=0}^{\infty}b_mf(q^mx)\frac{\Lambda_q^{-1}}{1-\Lambda_q^{-1}} g\\
&=&B(f) \frac{\Lambda_q^{-1}}{1-\Lambda_q^{-1}} g,\]

\[\notag
(f \frac{\Lambda_q^{-1}}{1-\Lambda_q^{-1}} gB)_-&=&\sum_{m=0}^{\infty}(f \frac{\Lambda_q^{-1}}{1-\Lambda_q^{-1}} gb_m\La^m)_-\\ \notag
&=&\sum_{m=0}^{\infty}(f \frac{\Lambda_q^{-1}}{1-\Lambda_q^{-1}}\La^m g(q^{-m}x)b_m(q^{-m}x))_-\\ \notag
&=&\sum_{m=0}^{\infty}f( \frac{\Lambda_q^{m-1}}{1-\Lambda_q^{-1}})_- g(q^{-m}x)b_m(q^{-m}x)\\ \notag
&=&\sum_{m=0}^{\infty}f \frac{\Lambda_q^{-1}}{1-\Lambda_q^{-1}} b_m(q^{-m}x)g(q^{-m}x)\\
&=&f \frac{\Lambda_q^{-1}}{1-\Lambda_q^{-1}}B^*(g).\]
Similar proof for the eq.\eqref{Cpos} can be got easily.

\end{proof}

Similarly as in \cite{EMTH,EZTH,Hedeterminant,rogueHMB}, we can get the $n$-fold Darboux transformation in the following theorem which will be used to generate new solutions.

\begin{theorem}\label{ndarboux}
The $n$-fold  Darboux transformation of EQTH equation is as following
\[W_n=1+t_1^{[n]}\Lambda_q^{-1}+t_2^{[n]}\Lambda_q^{-2}+\dots+t_{n}^{[n]}\Lambda_q^{-n}\]
where

\[ W_n\cdot\phi_{i}|_{i\leq n}=0.\]

The Darboux transformation leads to new solutions form seed solutions
\[u^{[n]}&=&u+(\Lambda_q-1)t_1^{[n]},\\
v^{[n]}&=&t_n^{[n]}(x)(\Lambda_q^{-n}v)t_n^{[n]-1}(q^{-1}x).\]
where
\[\notag &&W_n=\frac{1}{\Delta_n}\left|\begin{matrix}\begin{smallmatrix}
1& \Lambda_q^{-1}&\Lambda_q^{-2}&\dots & \Lambda_q^{-n}\\
\phi_1& \phi_1(q^{-1}x)&\phi_1(q^{-2}x)&\dots &\phi_1(q^{-n}x)\\
\phi_2& \phi_2(q^{-1}x)&\phi_2(q^{-2}x)&\dots &\phi_2(q^{-n}x)\\
\dots &\dots&\dots & \dots&\dots \\
\phi_n&\phi_n(q^{-1}x)&\phi_n(q^{-2}x)&\dots &\phi_n(q^{-n}x)\end{smallmatrix}\end{matrix}
\right|,\]
\[\notag&&\Delta_n=\left|\begin{matrix}\begin{smallmatrix}
\phi_1(q^{-1}x)&\phi_1(q^{-2}x)&\dots &\phi_1(q^{-n}x)\\
\phi_2(q^{-1}x)&\phi_2(q^{-2}x)&\dots &\phi_2(q^{-n}x)\\
\dots&\dots & \dots&\dots \\
\phi_n(q^{-1}x)&\phi_n(q^{-2}x)&\dots &\phi_n(q^{-n}x)\end{smallmatrix}\end{matrix}
\right|.\]
\end{theorem}
It can be easily checked that $W_n\phi_i=0,\ i=1,2,\dots,n.$

Taking seed solution $u=0,v=1$, then using Theorem \ref{ndarboux},  one can get the $n$-th new solution of the EQTH as

\[u^{[n]}&=&(1-\Lambda_q^{-1})\d_{t_{0}}\log \bar W_r(\phi_1,\phi_2,\dots\phi_n),\\
v^{[n]}&=&e^{(1-\Lambda_q^{-1})(1-\Lambda_q^{-1})\log \bar W_r(\phi_1,\phi_2,\dots\phi_n)},\]

where $\bar W_r(\phi_1,\phi_2,\dots\phi_n)$ is the q-deformed ``Wronskian"
\[\bar W_r(\phi_1,\phi_2,\dots\phi_n)=det (\Lambda_q^{-j+1} \phi_{n+1-i})_{1\leq i,j\leq n}.\]

\section{Multicomponent $q$-Toda hierarchy }

\subsection{Factorization Problem}

In this section, we will denote $G_N$ as a group which contains invertible elements of complex $N\times N$
complex matrices and denote its Lie algebra  $\g_N$ as the associative algebra  of complex $N\times N$
complex matrices $M_N(\C)$.

 Now we
introduce  the following free operators $ W_{N0},\bar  W_{N0}\in G_N$
\begin{align}
 \label{def:E}  W_{N0}&:=\sum_{k=1}^NE_{kk}\Exp{\sum_{j=0}^\infty
 t_{jk}\Lambda_q^{j}}, \\
\label{def:barE}   \bar W_{N0}&:=\sum_{k=1}^NE_{kk}\Exp{\sum_{j=0}^\infty\bar
   t_{j k}\Lambda_q^{-j}},
\end{align}
where $t_{jk}, \bar t_{jk} \in \C$
will play the role of continuous times.
 We   define the dressing operators $W_N,\bar W_N$ as follows
\begin{align}
\label{def:baker}W_N&:=S_N\cdot W_{N0},& \bar W_N&:=\bar S_N\cdot \bar  W_{N0}.
\end{align}
Given an element $g_N\in G_N$ and time series $t=(t_{jk}), \bar t=(\bar t_{jk}), s=(s_{j}); j,k \in \N, 1\leq k \leq N$, one can consider the factorization problem in $G_N$ \cite{manasinverse}
\begin{gather}
  \label{facW2}
  W_N\cdot g_N=\bar W_N,
\end{gather}
i.e.
 the factorization problem
\begin{gather}
  \label{factorization}
  S_N(t,\bar t,s)\cdot W_{N0}\cdot g_N=\bar S_N(t,\bar t,s)\cdot\bar W_{N0},\quad S_N\in G_{N-}\text{ and } \bar S_N\in G_{N+}.
\end{gather}
Observe that  $S_N,\bar S_N$ have expansions of the form
\begin{gather}
\label{expansion-S}
\begin{aligned}
S_N&=\I_N+\beta_1(x)\Lambda_q^{-1}+\beta_2(x)\Lambda_q^{-2}+\cdots\in G_{N-},\\
\bar S_N&=\bar\beta_0(x)+\bar\beta_1(x)\Lambda_q+\bar\beta_2(x)\Lambda_q^{2}+\cdots\in
G_{N+}.
\end{aligned}
\end{gather}

Also the inverse operators $S_N^{-1},\bar S_N^{-1}$ of operators $S,\bar S_N$ have expansions of the form
\begin{gather}
\begin{aligned}
S_N^{-1}&=\I_N+\beta'_1(x)\Lambda_q^{-1}+\beta'_2(x)\Lambda_q^{-2}+\cdots\in G_{N-},\\
\bar S_N^{-1}&=\bar\beta'_0(x)+\bar\beta'_1(x)\Lambda_q+\bar\beta'_2(x)\Lambda_q^{2}+\cdots\in
G_{N+}.
\end{aligned}
\end{gather}

 The Lax  operators $L,C_{kk},\bar C_{kk}\in\g_N$
 are defined by
\begin{align}
\label{Lax}  L&:=W_N\cdot\Lambda_q\cdot W_N^{-1}=\bar W_N\cdot\Lambda_q^{-1}\cdot \bar W_N^{-1}, \\
\label{C} C_{kk}&:=W_N\cdot E_{kk}\cdot W_N^{-1},& \bar C_{kk}&:=\bar
W_N\cdot E_{kk}\cdot \bar W_N^{-1},
\end{align}
and
have the following expansions
\begin{gather}\label{lax expansion}
\begin{aligned}
 L&=\Lambda_q+u(x)+v(x)\Lambda_q^{-1}, \\
C_{kk}&=E_{kk}+C_{kk,1}(x)\Lambda_q^{-1}+C_{kk,2}(x)\Lambda_q^{-2}+\cdots,\\
\bar \Cc_{kk}&=\bar C_{kk,0}(x)+\bar C_{kk,1}(x)\Lambda_q+\bar
C_{kk,2}(x)\Lambda_q^{2}+\cdots.
\end{aligned}
\end{gather}
 In fact the Lax  operators $L,C_{kk},\bar C_{kk}\in\g_N$
 can also be equivalently defined by
\begin{align}
\label{Lax}  L&:=S_N\cdot\Lambda_q\cdot S_N^{-1}=\bar S_N\cdot\Lambda_q^{-1}\cdot \bar S_N^{-1}, \\
\label{C} C_{kk}&:=S_N\cdot E_{kk}\cdot S_N^{-1},& \bar C_{kk}&:=\bar S_N\cdot E_{kk}\cdot \bar S_N^{-1}.
\end{align}

\section{ Lax equations of MQTH}

In this section we will use the factorization problem \eqref{facW2} to derive  Lax equations.
Let us first introduce some convenient notations.
\begin{definition}The matrix operators $C_{kk},\bar C_{kk},B_{jk},\bar B_{jk}$ are defined as follows
\begin{align}\label{satoS}
\begin{aligned}
C_{kk}&:=W_NE_{kk}W_N^{-1},\ \ \bar C_{kk}:=\bar W_N E_{kk}\bar W_N^{-1},\\
B_{jk}&:=W_NE_{kk}\Lambda_q^jW_N^{-1},\ \ \bar B_{jk}:=\bar W_N E_{kk}\Lambda_q^{-j}\bar W_N^{-1}.
\end{aligned}
\end{align}
\end{definition}

Now we give the definition of the  multicomponent $q$-Toda hierarchy(MQTH).
\begin{definition}The  multicomponent $q$-Toda hierarchy is a hierarchy in which the dressing operators $S_N,\bar S_N$ satisfy following Sato equations
\begin{align}
\label{satoSt} \epsilon \partial_{ t_{jk}}S_N&=-(B_{jk})_-.S_N,& \epsilon \partial_{ t_{jk}}\bar S_N&=(B_{jk})_+\cdot\bar S_N,  \\
\label{satoSbart}
\epsilon \partial_{ \bar t_{jk}}S_N&=-(\bar B_{jk})_-\cdot S_N,& \epsilon \partial_{ \bar t_{jk}}\bar S_N&=(\bar B_{jk})_+\cdot\bar S_N.\end{align}
\end{definition}
Then one can easily get the following proposition about $W_N,\bar W_N.$

\begin{proposition}The wave operators $W_N,\bar W_N$ satisfy following Sato equations
\begin{align}
\label{Wjk} \epsilon \partial_{ t_{jk}}W_N&=(B_{jk})_+\cdot W_N,& \epsilon \partial_{ t_{jk}}\bar W_N&=(B_{jk})_+\cdot\bar W_N,  \\
\label{Wbjk}\epsilon \partial_{ \bar t_{jk}}W_N&=-(\bar B_{jk})_-\cdot W_N,& \epsilon \partial_{ \bar t_{jk}}\bar W_N&=-(\bar B_{jk})_-\cdot\bar W_N.  \end{align}
\end{proposition}

 From the previous proposition we can derive the following  Lax equations for the Lax operators.
\begin{proposition}\label{Lax}
 The  Lax equations of the MQTH are as follows
   \begin{align}
\label{laxtjk}
  \epsilon \partial_{ t_{jk}} L&= [(B_{jk})_+,L],&
 \epsilon \partial_{ t_{jk}} C_{ss}&= [(B_{jk})_+,C_{ss}],&\epsilon \partial_{ t_{jk}} \bar C_{ss}&= [(B_{jk})_+,\bar C_{ss}],
\\
  \epsilon \partial_{ \bar t_{jk}} L&= [ (\bar B_{jk})_+,L],&
 \epsilon \partial_{ \bar t_{jk}} C_{ss}&= [(\bar B_{jk})_+,C_{ss}],&\epsilon \partial_{ \bar t_{jk}} \bar C_{ss}&= [(\bar B_{jk})_+,\bar C_{ss}].
\end{align}
\end{proposition}

To see this kind of hierarchy more clearly, the  multicomponent $q$-Toda equations as the $\d_{t_{1k}}$ flow equations  will be given in the next subsection.
\subsection{The multicomponent $q$-Toda equations}
 As a consequence of the factorization problem \eqref{facW2} and  Sato equations, after taking into account that   $S_N\in G_{N-}$ and $\bar S_N\in G_{N+}$
and using  the notation $\Exp{\phi_N}:=\bar\beta_0$ in $\bar S_N$, $B_{1k}$ has following form
\begin{gather}\label{exp-omega}
\begin{aligned}
B_{1k}&=E_{kk}\Lambda_q+U_k+V_k\Lambda_q^{-1},\ \ 1\leq k\leq N,
  \end{aligned}
\end{gather}
and we have the alternative expressions
\begin{gather}\label{exp-omega1}
\begin{aligned}
  U_k&:=\beta_1(x)E_{kk}-E_{kk}\beta_1(q x)=\epsilon \partial_{ t_{1k}}(\Exp{\phi_N(x)})\cdot\Exp{-\phi_N(x)},\\
 V_k&= \Exp{\phi_N(x)}E_{kk}\Exp{-\phi_N(q^{-1} x)}=-\epsilon \partial_{ t_{1k}}\beta_1(x).
\end{aligned}
\end{gather}

From Sato equations we deduce the following set of nonlinear
partial differential-difference equations
\begin{align}\left\{
\begin{aligned}
 \beta_1(x)E_{kk}-E_{kk}\beta_1(q x)&=\epsilon \partial_{ t_{1k}}(\Exp{\phi_N(x)})\cdot\Exp{-\phi_N(x)},\\
\partial_{t_{1k}}\beta_1(x)&=-\Exp{\phi_N(x)}E_{kk}\Exp{-\phi_N(q^{-1} x)}.\end{aligned}\right.
\label{eq:multitoda}
\end{align}
These equations constitute what we call the multicomponent $q$-Toda equations. Observe that if we cross the two equations in \eqref{eq:multitoda}, then we get
\begin{align*}
  \epsilon^2  \partial_{t_{1 k}}\big(\partial_{t_{1k}}(\Exp{\phi_N(x)})\cdot\Exp{-\phi_N(x)}\big)=
  E_{kk}\Exp{\phi_N(q x)}E_{kk}\Exp{-\phi_N(x)}-\Exp{\phi_N(x)}E_{kk}\Exp{-\phi_N(q^{-1} x)}E_{kk},
\end{align*}
which is the matrix extension of the following Toda equation (the case when $N=1$)
\begin{align*}
 \epsilon^2 \partial_{t_{11}}\partial_{t_{11}}(\phi_N(x))=\Exp{\phi_N(q x)-\phi_N(x)}-
 \Exp{\phi_N(x)-\phi_N(q^{-1} x)}.
\end{align*}

Besides above multicomponent $q$-Toda equations, the logarithmic flows the MQTH also contains some extended flow equations in the next subsection.

\section{Bi-Hamiltonian structure and tau symmetry}
To describe the integrability of the MQTH with the matrix-valued Lax operator
 \[L=\Lambda_q+u+v\Lambda_q^{-1},\ \ i.e. \ \ L_{ij}=\delta_{ij}\Lambda_q+u_{ij}+v_{ij}\Lambda_q^{-1},\] we will construct the bi-Hamiltonian structure and tau symmetry of the MQTH in this section.
For a matrix $A=(a_{ij})$, the vector field $\d_A$ over MQTH is defined by
\[\d_A=\sum_{i,j=1}^N\sum_{k\geq 0}a_{ij}^{(k)}(\frac{\d}{\d u_{ij}^{(k)}}+\frac{\d}{\d v_{ij}^{(k)}})=Tr \sum_{k\geq 0}A^{(k)}(\frac{\d}{\d u^{(k)}}+\frac{\d}{\d v^{(k)}}),\]
where
\[(\frac{\d}{\d u^{(k)}})_{ji}=\frac{\d}{\d u_{ij}^{(k)}},\ (\frac{\d}{\d v^{(k)}})_{ji}=\frac{\d}{\d v_{ij}^{(k)}}.\]

For two functionals $\bar f=\int f dx,\bar g=\int g dx$, we have
\[\d_A \bar f=\int\sum_{i,j=1}^N\sum_{k\geq 0}a_{ij}^{(k)}(\frac{\d f}{\d u_{ij}^{(k)}}+\frac{\d f}{\d v_{ij}^{(k)}}) dx=\int Tr \sum_{k\geq 0}A^{(k)}(\frac{\delta f}{\delta u^{(k)}}+\frac{\delta f}{\delta v^{(k)}}) dx.\]
Then we can define the hamiltonian bracket as
\[\{\bar f,\bar g\}=\int \sum_{w,w'}\frac{\delta f}{\delta w}\{w,w'\}\frac{\delta g}{\delta w'} dx,\ \ w, w'=u_{ij}\ or\ v_{ij}, \ 1\leq i,j \leq N.\]
The bi-Hamiltonian structure for the
MQTH can be given by the following two compatible Poisson brackets which is a generalization in matrix forms of the extended Toda hierarchy in \cite{CDZ}

\begin{eqnarray}
\{u(x)_{ij},u(y)_{pq}\}_1&=&\frac{1}{\epsilon} [\delta_{iq}u_{pj}(x)-\delta_{jp}u_{iq}(x)]\delta(x-y),\label{toda-pb1uu}\\
\{u(x)_{ij},v(y)_{pq}\}_1&=&\frac{1}{\epsilon} \left[\delta_{iq}\Lambda_q v_{pj}(x)-\delta_{jp}v_{iq}(x)\right]\delta(x-y),\label{toda-pb1uv}\\
\{v(x)_{ij},v(y)_{pq}\}_1&=&0,\label{toda-pb1vv}
\end{eqnarray}
\begin{eqnarray}\notag
\{u(x)_{ij},u(y)_{pq}\}_2&=&\frac{1}{\epsilon} \left[\delta_{iq}\Lambda_q v_{pj}(x)-\delta_{jp}v_{iq}(x)\Lambda_q^{-1}+\delta_{iq}\sum_{s=1}^Nu_{sj}\frac{\La_q}{\La_q-1}u_{ps}-u_{pj}\frac{\La_q}{\La_q-1}u_{iq}\right.\\
&&\left.-u_{iq}(\La_q-1)^{-1}u_{pj}+\delta_{jp}\sum_{s=1}^Nu_{is}(\La_q-1)^{-1}u_{sq}\right]\delta(x-y), \label{toda-pb2uu}\\ \notag
 \{ u(x)_{ij}, v(y)_{pq}\}_2 &=& \frac{1}{\epsilon} \left[\delta_{iq}\sum_{s=1}^Nu_{sj}\La_q^2 (\La_q-1)^{-1} v_{ps}(x)-u_{pj}\La_q (\La_q-1)^{-1}v_{iq}\right.\\
&&\left.-u_{iq}(\La_q-1)^{-1}\La_q v_{pj}(x)+\delta_{jp}\sum_{s=1}^Nu_{is}(\La_q-1)^{-1}v_{sq}\right]
\delta(x-y),\label{toda-pb2uv}\\ \notag
 \{ v(x)_{ij}, v(y)_{pq}\}_2 &=& {1\over \epsilon} \left[\delta_{iq}\sum_{s=1}^Nv_{sj}\La_q^2 (\La_q-1)^{-1}v_{ps}(x)-v_{pj}\La_q (\La_q-1)^{-1}v_{iq}\right.\\
&&\left.-v_{iq}(\La_q-1)^{-1}v_{pj}(x)+\delta_{jp}\sum_{s=1}^Nv_{is}\La_q^{-1}(\La_q-1)^{-1}v_{sq}\right]\delta(x-y). \label{toda-pb2vv}
\end{eqnarray}

In the following theorem, we will prove the above poisson structures can be considered as the bi-Hamiltonian structure of the MQTH.
\begin{theorem}
The flows of the MQTH  are Hamiltonian systems
of the form
\[
\frac{\d u_{pq}}{\d t_{j,k}}&=&\{u_{pq},H_{j,k}\}_1, \  \frac{\d v_{pq}}{\d t_{j,k}}=\{v_{pq},H_{j,k}\}_1,\\
\frac{\d u_{pq}}{\d \bar t_{j,k}}&=&\{u_{pq},\bar H_{j,k}\}_1, \  \frac{\d v_{pq}}{\d\bar t_{j,k}}=\{v_{pq},\bar H_{j,k}\}_1,
\quad k=0,1,\dots N;\ j\ge 0.
\label{td-ham}
\]
They satisfy the following bi-Hamiltonian recursion relation
\[
\{\cdot,H_{n-1,k}\}_2&=&\{\cdot,H_{n,k}\}_1,\ \{\cdot,\bar H_{n-1,k}\}_2=
\{\cdot,\bar H_{n,k}\}_1.
\]
Here the Hamiltonians have the form
\begin{equation}
F_{j,k}=\int f_{j,k}(u,v; u_x,v_x; \dots; \epsilon) dx,
\end{equation}
with the Hamiltonian $F_{j,k}=H_{j,k},\bar H_{j,k}$ and the Hamiltonian densities $f_{j,k}=h_{j,k},\bar h_{j,k}$ given by
\[
 h_{j,k}&=&Tr Res \, C_{kk}L^{j},\ \bar h_{j,k}=Tr Res \, \bar C_{kk} L^{j}.
\]

\end{theorem}

For readers' convenience, now we will write down the first several Hamiltonian densities explicitly as follows
\[ h_{0,k}&=&Tr Res\,  C_{kk}=Tr E_{kk}=1,\\
 h_{1,k}&=&Tr Res \, C_{kk}L=Tr [(1-\La_q)^{-1}uE_{kk}-E_{kk}\frac{\La_q}{1-\La_q}u]=u_{kk},\\
 \ \bar h_{0,k}&=&Tr Res\,\bar C_{kk} =Tr\, \tilde \beta_0E_{kk} \tilde \beta_0^{-1}  \,  ,\\
 \ \bar h_{1,k}&=&Tr Res \, \bar C_{kk} L=Tr\, (\tilde \beta_1E_{kk}\tilde \beta_0^{-1}-\tilde \beta_0E_{kk}\tilde \beta_0^{-1}(q^{-1} x)\tilde \beta_1(q^{-1} x)\tilde \beta_0^{-1}),
\]
with
\[\tilde \beta_0=v\tilde \beta_0(q^{-1} x),\ \tilde \beta_1=u\tilde \beta_0+v\tilde \beta_1(q^{-1} x).\]
When $N=1$, the above conserved densities will be the ones of the  Toda  hierarchy in \cite{CDZ}.
Similarly as \cite{CDZ}, the tau symmetry of the MQTH can be proved in the  following theorem.
\begin{theorem}\label{tausymmetry}
The Hamiltonian densities  of the MQTH have the following tau-symmetry property:
\begin{equation}
\frac{\d h_{\alpha,m}}{\d t_{j,k}}=\frac{\d
h_{j,k}}{\d t_{\alpha,m}},\quad \frac{\d \bar h_{\alpha,m}}{\d t_{j,k}}=\frac{\d
 h_{j,k}}{\d\bar t_{\alpha,m}},
\end{equation}
\begin{equation}
\frac{\d h_{\alpha,m}}{\d \bar t_{j,k}}=\frac{\d
\bar h_{j,k}}{\d  t_{\alpha,m}},\quad \frac{\d \bar h_{\alpha,m}}{\d \bar t_{j,k}}=\frac{\d
\bar h_{j,k}}{\d \bar t_{\alpha,m}},
\end{equation}
\end{theorem}
\begin{proof} Let us prove the theorem for the  first equation,
other cases can be proved in a similar way
\[
&&\frac{\d  h_{m,s}}{\d t_{n,k}} =\,Tr Res[-(C_{kk}L^{n})_-, C_{ss}L^m ]\notag\\
&&=Tr Res[(C_{ss}L^m )_+,(C_{ss}L^{n})_-]\notag\\
&& =Tr Res[(C_{ss}L^m )_+,C_{kk}L^{n}]=\frac{\d h_{n,k}}{\d t_{m,s}}.
\]
\end{proof}

 This property justifies the following definition of the
tau function for the MQTH:

\begin{definition} The  $tau$ function $\tau_N$ of the MQTH can be defined by
the following expressions in terms of the densities of the Hamiltonians:
\begin{equation}
h_{j,n}=\epsilon (\Lambda_q-1)\frac{\d\log\tau_N}{\d t_{j,n}},
\end{equation}
\begin{equation}
\bar h_{j,n}=\epsilon (\Lambda_q-1)\frac{\d\log\tau_N}{\d\bar t_{j,n}}.
\end{equation}
\end{definition}

{\bf {Acknowledgements:}}
  Chuanzhong Li is supported by the National Natural Science Foundation of China under Grant No. 11201251, Zhejiang Provincial Natural Science Foundation of China under Grant No. LY12A01007, the Natural Science Foundation of Ningbo under Grant No. 2013A610105, 2014A610029 and K.C.Wong Magna Fund in
Ningbo University.

\end{document}